\documentclass[journal]{IEEEtran}
\usepackage{xspace,amsmath,amssymb,amsthm,amsfonts,epsfig,syntonly,dsfont}
\usepackage{cite,bm,color,url,textcomp}
\usepackage{algorithmicx,algorithm,algpseudocode}
\usepackage{epstopdf,cases}
\usepackage{empheq,float}
\usepackage{lettrine}
\usepackage{graphicx,subfigure,balance}
\newtheorem{assumption}{Assumption}
\newtheorem{proposition}{Proposition}

\newtheorem{lemma}{Lemma}

\newtheorem{theorem}{Theorem}

\newtheorem{remark}{Remark}

\long\def\symbolfootnote[#1]#2{\begingroup
\def\thefootnote{\fnsymbol{footnote}}
\footnote[#1]{#2}\endgroup}
\def\baselinestretch{0.97}
\psfull

\allowdisplaybreaks[2]

\begin{document}

\title{Learn-and-Adapt Stochastic Dual Gradients for Network Resource Allocation}
\author{Tianyi Chen,~\IEEEmembership{Student\! Member, IEEE}, Qing Ling,~\IEEEmembership{Senior\! Member, IEEE}, and Georgios B. Giannakis,~\IEEEmembership{Fellow, IEEE}\vspace{-0.3cm}
\thanks {Work in this paper was supported by NSF 1509040, 1508993, 1509005, NSF China 61573331, NSF Anhui 1608085QF130, and CAS-XDA06040602.}

\thanks{T. Chen and G. B. Giannakis are with the Department of Electrical and Computer Engineering and the Digital Technology Center, University of Minnesota, Minneapolis, MN 55455 USA. Emails: \{chen3827, georgios\}@umn.edu

Qing Ling is with the School of Data and Computer Science, Sun Yat-Sen University, Guangzhou, Guangdong 510006, China and also with the Department of Automation, University of Science and Technology of China, Hefei, Anhui 230026, China. Email: qingling@ieee.org
}
}

\maketitle

\vspace{-0.5cm}

\begin{abstract}
Network resource allocation shows revived popularity in the era
of data deluge and information explosion. Existing
stochastic optimization approaches fall short in attaining a desirable
cost-delay tradeoff. Recognizing the
central role of Lagrange multipliers in network resource
allocation, a novel learn-and-adapt stochastic
dual gradient (LA-SDG) method is developed in this paper to learn the sample-optimal
Lagrange multiplier from historical data, and accordingly adapt the upcoming 
resource allocation strategy. Remarkably, LA-SDG only requires
just an extra sample (gradient) evaluation relative to the celebrated
stochastic dual gradient (SDG) method. LA-SDG can be interpreted
as a foresighted learning scheme with \textit{an eye on the
future}, or, a modified heavy-ball iteration from an optimization
viewpoint. It is established - both theoretically and empirically
- that LA-SDG markedly improves the cost-delay tradeoff
over state-of-the-art allocation schemes.

\end{abstract}


\section{Introduction}\label{S:Intro}

In the era of big data analytics, cloud computing and Internet of
Things, the growing demand for massive data processing challenges
existing resource allocation approaches. 
Huge volumes of data acquired by distributed sensors in the presence of 
operational uncertainties caused by, e.g., renewable energy, call for scalable and adaptive network control schemes.
Scalability of a desired approach refers to low
complexity and amenability to distributed implementation, while
adaptivity implies capability of online adjustment to dynamic
environments.

Allocation of network resources can be traced back to the
seminal work of \cite{tassiulas1992}. 
Since then, popular allocation algorithms operating in the dual domain are first-order methods based on dual gradient ascent, either deterministic \cite{low1999} or stochastic \cite{Geor06,neely2010}.
Thanks to their simple computation and implementation, these approaches
have attracted a great deal of recent interest, and have been
successfully applied to cloud, transportation and power grid
networks; see, e.g., \cite{chen2016,chen2016jsac,gregoire2015,sun2016}.
However, their major limitation is \emph{slow convergence}, which results in high \emph{network delay}. 
Depending on the application domain, the
delay can be viewed as workload queuing time in a cloud network,
traffic congestion in a transportation network, or energy level of batteries in a power network. To address this delay issue, recent attempts aim at accelerating first- and second-order optimization algorithms
\cite{beck2014,liu2016,wei2013,zargham2013}. Specifically,
momentum-based accelerations over first-order methods were investigated
using Nesterov \cite{beck2014}, or, heavy-ball iterations \cite{liu2016}. 
Though these approaches work well in
static settings, their performance degrades with online scheduling, as evidenced by the increase in accumulated steady-state error \cite{yuan2016}. 
On the other hand, second-order methods such as the decentralized quasi-Newton approach and its dynamic variant developed 
in \cite{wei2013} and \cite{zargham2013}, incur high overhead to compute and communicate the 
decentralized Hessian approximations.

Capturing prices of resources, Lagrange multipliers play a central role in stochastic resource allocation algorithms \cite{huang2011}.
Given abundant historical data
in an online optimization setting, a natural question arises:
\textit{Is it possible to learn the optimal prices from past
data, so as to improve the performance of online resource allocation
strategies?} 
The rationale here is that past data contain statistics of network states, and learning from them can aid coping with the stochasticity of future resource allocation.
A recent work in this direction is
\cite{huang2014}, which considers resource allocation with a \textit{finite} number of possible network states and allocation actions.
The learning procedure, however,
involves constructing a histogram to estimate the underlying
distribution of the network states, and explicitly
solves an empirical dual problem. While
constructing a histogram is feasible for a probability
distribution with finite support, quantization
errors and prohibitively high complexity are
inevitable for a continuous distribution with infinite support.

%

In this context, the present paper aims to design a novel online resource allocation algorithm that leverages online
learning from historical data for stochastic optimization of the ensuing
allocation stage. The resultant approach, which we term 
``learn-and-adapt'' stochastic dual gradient (LA-SDG) method,
only doubles computational complexity of the classic
stochastic dual gradient (SDG) method. 
With this minimal cost, LA-SDG mitigates steady-state oscillation, 
which is common in stochastic first-order acceleration methods \cite{yuan2016,liu2016}, while avoiding computation of the Hessian approximations present in the second-order
methods \cite{wei2013,zargham2013}. Specifically, LA-SDG only requires one more past sample to compute an extra stochastic dual gradient, in contrast to constructing
costly histograms and solving the resultant large-scale problem \cite{huang2014}.

The main contributions of this paper are summarized next.
\begin{enumerate}
\item [c1)] Targeting a low-complexity online solution, LA-SDG
only takes an additional dual gradient step relative to the classic SDG iteration. 
This step enables adapting the resource
allocation strategy through learning from historical data.
Meanwhile, LA-SDG is linked with the stochastic heavy-ball method, nicely inheriting its fast
convergence in the initial stage, while reducing its steady-state oscillation.
 \item [c2)]
The novel LA-SDG approach,
parameterized by a positive constant $\mu$, provably yields an attractive 
cost-delay tradeoff $[\mu,\log^2(\mu)/\sqrt{\mu}]$, which
improves upon the standard tradeoff $[\mu,{1}/{\mu}]$ of the
SDG method \cite{neely2010}. Numerical tests further
corroborate the performance gain of LA-SDG over existing resource allocation schemes.

\end{enumerate}


\emph{Notation}. $\mathbb{E}$ denotes the expectation operator, 
$\mathbb{P}$ stands for probability; $(\cdot)^{\top}$ stands for
vector and matrix transposition, and $\|\mathbf{x}\|$ denotes the
$\ell_2$-norm of a vector $\mathbf{x}$. Inequalities for vectors,
e.g., $\mathbf{x} > \mathbf{0}$, are defined entry-wise. The
positive projection operator is defined as $[a]^+:=\max\{a,0\}$,
also entry-wise.

\section{Network Resource Allocation}\label{S:ModelPrelim}
In this section, we start with a generic network model and its
resource allocation task in Section \ref{subsec.PF}, and then
introduce a specific example of resource allocation in cloud
networks in Section \ref{subsec.exp}. The proposed approach is
applicable to more general network resource allocation tasks such as geographical load
balancing in cloud networks \cite{chen2016}, traffic control in
transportation networks \cite{gregoire2015}, and energy management
in power networks \cite{sun2016}.

\subsection{A unified resource allocation model}\label{subsec.PF}

Consider discrete time $t\in\mathbb{N}$, and a network represented as a directed graph ${\cal G}=({\cal
I},\,{\cal E})$ with nodes ${\cal I}:=\{1,\ldots,I\}$ and edges ${\cal E}:=\{1,\ldots,E\}$. 
Collect the workloads across edges $e=(i,j)\in {\cal E}$ in a resource allocation
vector $\mathbf{x}_t\in \mathbb{R}^{E}$. The $I\times E$ node-incidence matrix is formed with the $(i,e)$-th entry
   \begin{equation}\label{eq.incidence}
    \mathbf{A}_{(i,e)}=
    \left\{
    \begin{array}{rl}
         {1,}~  &\text{if link $e$ enters node $i$}\\
         {-1,}~ &\text{if link $e$ leaves node $i$}\\
         {0,}~ &\text{else.}
    \end{array}
   \right.
\end{equation}
We assume that each row of
$\mathbf{A}$ has at least one $-1$ entry, and each column of
$\mathbf{A}$ has at most one $-1$ entry, meaning that each
node has at least one outgoing link, and each link has at most one
source node. With $\mathbf{c}_t\in \mathbb{R}_+^{I}$ collecting the randomly arriving workloads of all nodes per slot $t$, the aggregate (endogenous plus exogenous) workloads of all nodes are
$\mathbf{A}\mathbf{x}_t+\mathbf{c}_t$. If the $i$-th entry of
$\mathbf{A}\mathbf{x}_t+\mathbf{c}_t$ is positive, there is
service residual queued at node $i$; otherwise, node $i$ over-serves the current arrival. With a
workload queue per node, the queue length vector $\mathbf{q}_t:=[q_t^1,\ldots,q_t^I]^{\top}\in \mathbb{R}_+^{I}$ obeys the recursion 
\begin{equation}
	\mathbf{q}_{t+1}=\left[\mathbf{q}_t+\mathbf{A}\mathbf{x}_t+\mathbf{c}_t\right]^{+}\!,~\forall t 
\end{equation}
where $\mathbf{q}_t$ can represent the amount of user requests buffered in data queues, or energy stored in batteries, and $\mathbf{c}_t$ is the corresponding exogenously arriving workloads or harvested renewable energy of all nodes per slot $t$.
Defining $\Psi_t(\mathbf{x}_t) :=
\Psi(\mathbf{x}_t;\bm{\phi}_t)$ as the aggregate network cost parameterized by the random vector
$\bm{\phi}_t$, the local cost per node $i$ is $\Psi_t^i(\mathbf{x}_t):=
\Psi^i(\mathbf{x}_t;\bm{\phi}_t^i)$, and $\Psi_t(\mathbf{x}_t)=\sum_{i\in{\cal I}}\Psi_t^i(\mathbf{x}_t)$.
The model here is quite general. 
The duration of time slots can vary from (micro-)seconds in cloud networks, minutes in road networks, to even hours in power networks; the nodes can present the distributed front-end mapping nodes and back-end data centers in cloud networks, intersections in traffic networks, or, buses and substations in power networks; the links can model wireless/wireline channels, traffic lanes, and power transmission lines; while the resource vector $\mathbf{x}_t$ can include the size of data workloads, the number of vehicles, or the amount of energy.

Concatenating the random
parameters into a random state vector
$\mathbf{s}_t:=[\bm{\phi}_t^{\top},\mathbf{c}_t^{\top}]^{\top}$, the resource allocation task is to determine
the allocation $\mathbf{x}_t$ in response to the observed (realization) $\mathbf{s}_t$ ``on the fly,'' so
as to minimize the \textit{long-term average} network cost subject
to queue stability at each node, and operation feasibility at each
link. Concretely, we have
\begin{subequations}
\label{eq.prob}
\begin{align}
{\Psi}^{*}:=&\min_{\{\mathbf{x}_t,\forall t\}}\, \lim_{T\rightarrow \infty}\frac{1}{T}\sum_{t=1}^T \mathbb{E}\left[\Psi_t(\mathbf{x}_t)\right]  \label{eq.proba}\\
\text{s.t.}~~~ &\mathbf{q}_{t+1}=\left[\mathbf{q}_t+\mathbf{A}\mathbf{x}_t+\mathbf{c}_t\right]^{+}\!,~\forall t\label{eq.probm}\\
&\lim_{T\rightarrow \infty} \frac{1}{T} \sum_{t=1}^{T} \mathbb{E}\left[\mathbf{q}_t\right]< \infty \label{eq.probn}\\
&~\mathbf{x}_t \in {\cal X}:=\{\mathbf{x}\,|\,\mathbf{0}\leq
\mathbf{x}\leq \bar{\mathbf{x}}\},~\forall t\label{eq.probl}
\end{align}
\end{subequations}
where ${\Psi}^{*}$ is the optimal objective of problem \eqref{eq.prob}, which includes also future information; $\mathbb{E}$ is taken over $\mathbf{s}_t:=[\bm{\phi}_t^{\top},\mathbf{c}_t^{\top}]^{\top}$ as well as possible randomness of optimization
variable $\mathbf{x}_t$; constraints \eqref{eq.probn} ensure queue stability\footnote{Here we focus on the strong stability
given by \cite[Definition 2.7]{neely2010}, which requires the
time-average expected queue length to be finite.}; and \eqref{eq.probl} confines the instantaneous allocation variables to stay within a time-invariant box constraint set ${\cal X}$, which is
specified by, e.g., link capacities, or, server/generator capacities.

The queue dynamics in \eqref{eq.probm} couple the
optimization variables over an infinite time horizon, which implies that the decision variable at the current slot will have effect on all the future decisions.  
Therefore,
finding an optimal solution of \eqref{eq.prob} calls for dynamic
programming \cite{borkar2002}, which is known to suffer from the ``curse of
dimensionality'' and intractability in an online setting. 
In Section
\ref{subsec.reform}, we will circumvent this obstacle by relaxing
\eqref{eq.probm}-\eqref{eq.probn} to limiting average constraints, 
and employing dual decomposition techniques.

\subsection{Motivating setup}\label{subsec.exp}

The geographic load
balancing task in a cloud network \cite{chen2016,Urg11,chen2017} takes the form of \eqref{eq.prob} 
with $J$ mapping nodes (e.g., DNS servers) indexed by ${\cal J}:=\{1,\ldots,J\}$, $K$ data centers indexed by ${\cal K}:=\{J+1,\ldots,J+K\}$. 
To match the definition in Section \ref{subsec.PF}, consider a virtual outgoing node (indexed by $0$) from each data center, and let $(k,0)$ represent this outgoing link.
Define further the node set ${\cal I}:={\cal J}\bigcup {\cal K}$ that includes all nodes except the virtual one, and the edge set ${\cal
E}:=\{(j,k),\forall j\!\in\!{\cal J},k\!\in\!{\cal K}\}\bigcup\{(k,0),\forall k\!\in\!{\cal K}\}$ that contains links connecting
mapping nodes with data centers, and outgoing links from data centers.

Per slot $t$, each mapping node
$j$ collects the amount of user data requests $c_t^j$, and forwards the amount
$x_t^{jk}$ on its link to data center $k$ constrained by the bandwidth
availability. Each data center $k$ schedules workload 
processing $x_t^{k0}$ according to its resource availability. 
The amount $x_t^{k0}$ can be also viewed as the resource on its virtual outgoing link $(k,0)$. 
The bandwidth limit of link $(j,k)$ is $\bar{x}^{jk}$, while the
resource limit of data center $k$ (or link $(k,0)$) is $\bar{x}_t^{k0}$. 
Similar to those in Section \ref{subsec.PF}, we have the optimization vector $\mathbf{x}_t:=\{x_t^{ij},\,\forall (i,j)\!\in\! {\cal E}\}\!\in\!\mathbb{R}^{|\cal E|}$,
$\mathbf{c}_t:=[c_t^1,\ldots,c_t^J,0\ldots,0]^{\top}\!\in\!\mathbb{R}^{J+K}$, and
$\bar{\mathbf{x}}:=\{\bar{x}_t^{ij},\,\forall (i,j)\!\in\! {\cal E}\}\!\in\!\mathbb{R}^{|\cal E|}$.
With these notational conventions, we have an $|{\cal I}|\times |{\cal E}|$ node-incidence matrix $\mathbf{A}$ as in \eqref{eq.incidence}. 
At each mapping node and data center, undistributed or unprocessed
workloads are buffered in queues obeying \eqref{eq.probm} with queue length $\mathbf{q}_t\in\mathbb{R}_+^{J+K}$; see also the system diagram in Fig. \ref{fig:system}.

\begin{figure}[t]
\hspace{-0.2cm}
\includegraphics[width=0.5\textwidth]{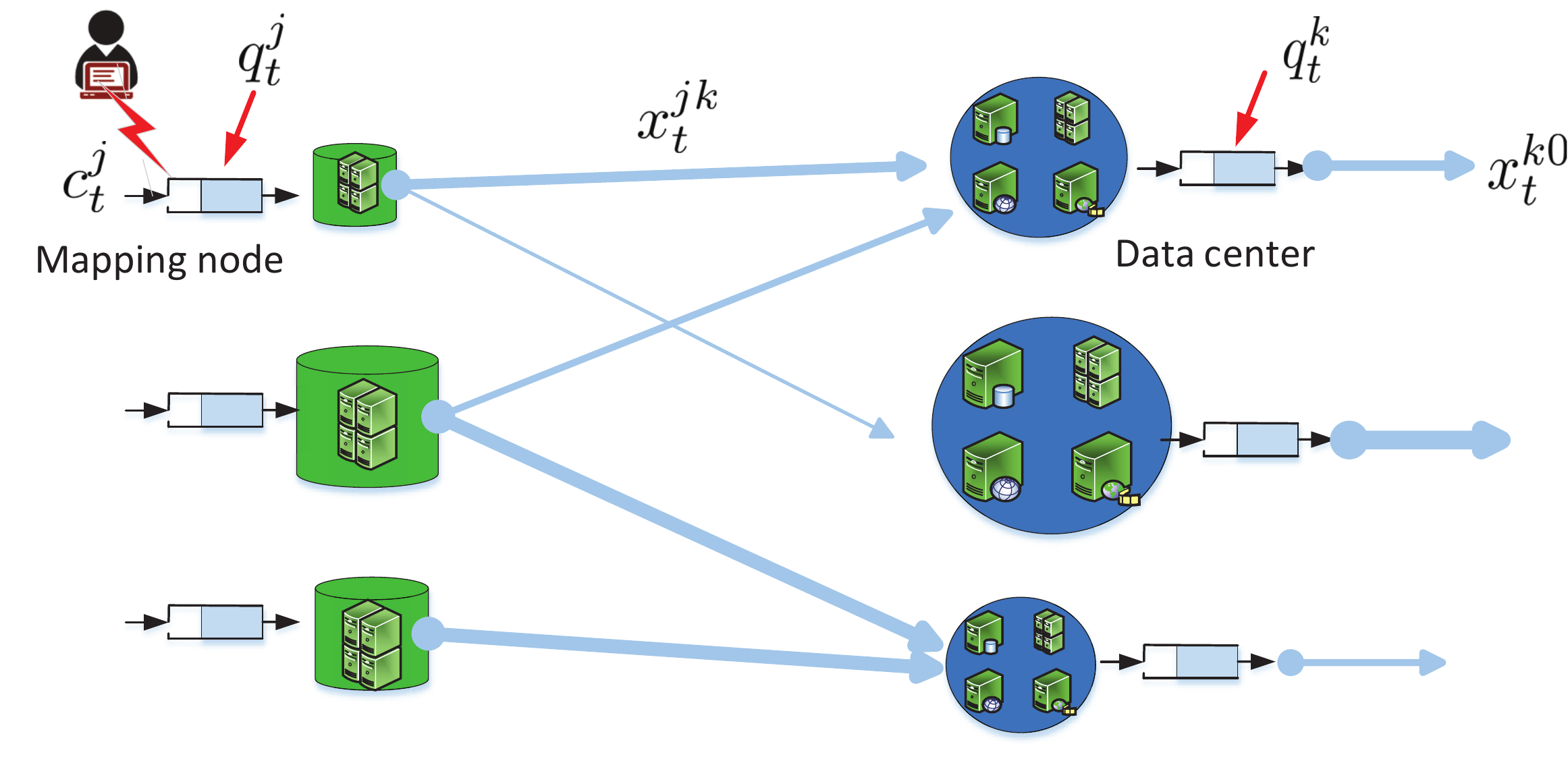}
\vspace{-0.5cm} \caption{A diagram of online geographical load balancing. Per time $t$, mapping node $j$ has an exogenous workload
$c_t^j$ plus that stored in the queue $q_t^j$, and schedules
workload $x_t^{jk}$ to data center $k$. Data center $k$ serves an amount of workload
$x_t^{k0}$ out of all the assigned $x_t^{jk}$ as well as that stored in
the queue $q_t^{k}$. The thickness of each edge is proportional to its capacity.} 
\vspace{-0.3cm} \label{fig:system}
\end{figure}

Performance is characterized by the aggregate cost of power consumed at the data centers plus the bandwidth costs at the mapping nodes, namely
\begin{equation}\label{eq.netcost}
\Psi_t(\mathbf{x}_t):=\sum_{k\in{\cal
K}}~\underbrace{~~\Psi_t^k(x_t^{k0})~~}_{\rm power~cost}~+~\sum_{j\in
{\cal J}}\sum_{k\in{\cal K}}\underbrace{~~\Psi_t^{jk}
(x_t^{jk})~~}_{\rm bandwidth~cost}.
\end{equation}
The power cost $\Psi_t^k (x_t^{k0}):=\Psi^k(x_t^{k0};\bm{\phi}_t^k)$,  parameterized by the random vector $\bm{\phi}_t^k$, 
captures the local marginal price, and the renewable generation at
data center $k$ during time period $t$. The bandwidth cost 
$\Psi_t^{jk} (x_t^{jk}):=\Psi^{jk} (x_t^{jk};\bm{\phi}_t^{jk})$,  parameterized by the random vector
$\bm{\phi}_t^{jk}$, 
characterizes the heterogeneous cost of data transmission due to
spatio-temporal differences. 
To match the unified model in Section II-A, the local cost at data center $k\in{\cal K}$ is its power cost $\Psi_t^k (x_t^{k0})$, and the local cost at mapping node $j\in{\cal J}$ becomes $\Psi_t^j (\{x_t^{jk}\}):=\sum_{k\in{\cal K}}\Psi_t^{jk}(x_t^{jk})$. Hence, the cost in \eqref{eq.netcost} can be also written as $\Psi_t(\mathbf{x}_t):=\sum_{i\in{\cal I}}\Psi_t^i (\mathbf{x}_t)$. 
Aiming to minimize the
time-average of \eqref{eq.netcost}, geographical load
balancing fits the formulation in \eqref{eq.prob}.

\section{Online Network Management via SDG}
In this section, the dynamic problem \eqref{eq.prob} is reformulated to a tractable form, and classical stochastic dual gradient (SDG) approach is revisited, along with a brief discussion of its online performance.

\subsection{Problem reformulation}\label{subsec.reform}
Recall in Section \ref{subsec.PF} that the main challenge of solving \eqref{eq.prob} resides in time-coupling constraints and unknown distribution of the underlying random processes. 
Regarding the first hurdle, combining \eqref{eq.probm} with \eqref{eq.probn}, it can be shown that in the long term, workload arrival and departure rates must
satisfy the following necessary condition \cite[Theorem 2.8]{neely2010}
\begin{equation}\label{Queue-relax}
\lim_{T\rightarrow \infty}\frac{1}{T}\sum_{t=1}^T\mathbb{E}\left[\mathbf{A}\mathbf{x}_t+\mathbf{c}_t\right]\leq \mathbf{0}
\end{equation}
given that the initial queue length is finite, i.e., $\|\mathbf{q}_1\|\leq \infty$. 
In other words, on average
all buffered delay-tolerant workloads should be served. Using
\eqref{Queue-relax}, a relaxed version of
\eqref{eq.prob} is
\begin{align} \label{eq.reform}
\tilde{\Psi}^{*} := \min_{\{\mathbf{x}_t,\forall t\}} \, \lim_{T\rightarrow
\infty}\frac{1}{T}\sum_{t=1}^T
\mathbb{E}\left[\Psi_t(\mathbf{x}_t)\right] ~~ \text{s.t.}~
\eqref{eq.probl}~\text{and}~\eqref{Queue-relax}
\end{align}
where $\tilde{\Psi}^{*}$ is the optimal objective for the relaxed problem \eqref{eq.reform}.

Compared to \eqref{eq.prob}, problem \eqref{eq.reform} eliminates the
time coupling across variables $\{\mathbf{q}_t,\forall t\}$ by replacing \eqref{eq.probm} and \eqref{eq.probn}
with \eqref{Queue-relax}. Since \eqref{eq.reform} is a
relaxed version of \eqref{eq.prob} with the optimal objective $\tilde{\Psi}^{*}\leq
{\Psi}^{*}$, 
if one solves \eqref{eq.reform} instead of \eqref{eq.prob}, it will be prudent to derive an optimality bound on ${\Psi}^{*}$, provided that the sequence of solutions $\{\mathbf{x}_t,\forall t\}$ obtained by solving \eqref{eq.reform} is feasible for the relaxed constraints \eqref{eq.probm} and \eqref{eq.probn}.
Regarding the relaxed problem \eqref{eq.reform}, using arguments
similar to those in \cite[Theorem 4.5]{neely2010}, it can be shown that
if the random state $\mathbf{s}_t$ is independent and identically
distributed (i.i.d.) over time $t$, there exists a
\textit{stationary} control policy $\bm{\chi}^*(\cdot)$, which is
a pure (possibly randomized) function of the realization of random state $\mathbf{s}_t$ (or the \textit{observed} state $\mathbf{s}_t$); i.e., it satisfies
\eqref{eq.probl}, as well as guarantees that
$\mathbb{E}[\Psi_t(\bm{\chi}^*(\mathbf{s}_t))] = \tilde{\Psi}^{*}$ and
$\mathbb{E}[\mathbf{A}\bm{\chi}^*(\mathbf{s}_t)+\mathbf{c}_t]\leq \mathbf{0}$.
As the optimal policy $\bm{\chi}^*(\cdot)$ is time invariant, it implies that the \textit{dynamic} problem
\eqref{eq.reform} is equivalent to the following time-invariant \textit{ensemble} program
\begin{subequations}\label{eq.reform2}
\begin{align}
\tilde{\Psi}^{*} := &\min_{\bm{\chi}(\cdot)} \; \mathbb{E}\left[\Psi\big(\bm{\chi}(\mathbf{s}_t);\mathbf{s}_t\big)\right]\label{eq.reforme0}\\
\text{s.t.}~~~&\mathbb{E}[\mathbf{A}\bm{\chi}(\mathbf{s}_t)+\mathbf{c}(\mathbf{s}_t)]\leq \mathbf{0} \label{eq.reforme1}\\
&\bm{\chi}(\mathbf{s}_t)  \in {\cal X},~\forall \mathbf{s}_t\in \mathcal{S}
\label{eq.reforme2}
\end{align}
\end{subequations}
where
$\bm{\chi}(\mathbf{s}_t)\!:=\!\mathbf{x}_t$, $\mathbf{c}(\mathbf{s}_t)=\mathbf{c}_t$, and
$\Psi\big(\bm{\chi}(\mathbf{s}_t);\mathbf{s}_t\big)\!:=\!\Psi_t(\mathbf{x}_t)$; set $\mathcal{S}$ is the sample space of $\mathbf{s}_t$, and the constraint \eqref{eq.reforme2} holds almost surely.
Observe that the index $t$ in \eqref{eq.reform2} can be dropped, since the expectation is taken over the distribution of random variable $\mathbf{s}_t$, which is time-invariant. 
Leveraging the equivalent form \eqref{eq.reform2}, the remaining task boils down to finding the optimal policy that achieves the minimal objective in \eqref{eq.reforme0} and obeys the constraints \eqref{eq.reforme1} and \eqref{eq.reforme2}.\footnote{Though there may exist other time-dependent policies that generate the optimal solution to \eqref{eq.reform}, our attention is restricted to the one that purely depends on the observed state $\mathbf{s}\in\mathcal{S}$, which can be time-independent \cite[Theorem 4.5]{neely2010}.}
Note that the optimization in \eqref{eq.reform2} is with respect to a
stationary policy $\bm{\chi}(\cdot)$, which is an infinite dimensional problem in the primal domain. However, there is a
finite number of expected constraints [cf.
\eqref{eq.reforme1}]. Thus, the dual problem contains a finite
number of variables, hinting to the effect that solving \eqref{eq.reform2} is tractable in the dual domain \cite{marques12,Ale10}.

\subsection{Lagrange dual and optimal policy}\label{subsec.Lag-dual}

With $\bm{\lambda}\in \mathbb{R}_+^{I}$ denoting the Lagrange
multipliers associated with \eqref{eq.reforme1},
the Lagrangian of \eqref{eq.reform2} is
\begin{equation}\label{eq.Lam0}
    {\cal L}(\bm{\chi},\bm{\lambda}) \!:=\mathbb{E}\big[{\cal
    L}_t(\mathbf{x}_t,\bm{\lambda})\big]
\end{equation}
with $\bm{\lambda}\geq \mathbf{0}$, and the instantaneous Lagrangian is
\begin{align}\label{eq.Lam}
{\cal L}_t(\mathbf{x}_t,\bm{\lambda})
\!:=&\Psi_t(\mathbf{x}_t)+\bm{\lambda}^{\top}(\mathbf{A}\mathbf{x}_t+\mathbf{c}_t)
\end{align}
where constraint \eqref{eq.reforme2} remains implicit.
Notice that the instantaneous objective $\Psi_t(\mathbf{x}_t)$ and the instantaneous constraint $\mathbf{A}\mathbf{x}_t+\mathbf{c}_t$ are both parameterized by the observed state $\mathbf{s}_t:=[\bm{\phi}_t^{\top},\mathbf{c}_t^{\top}]^{\top}$ at time $t$; i.e., ${\cal L}_t(\mathbf{x}_t,\bm{\lambda})={\cal L}(\bm{\chi}(\mathbf{s}_t),\bm{\lambda};\mathbf{s}_t)$.

Correspondingly, the Lagrange dual function is defined as the minimum of the Lagrangian over the
all feasible primal variables \cite{bertsekas2003}, given by
\begin{subequations}\label{eq.dual-func}
\begin{align}\label{eq.dual-func1}
{\cal D}(\bm{\lambda}):&=\min_{\{\bm{\chi}(\mathbf{s}_t)  \in {\cal X},~\forall \mathbf{s}_t\in \mathcal{S}\}}\,{\cal
L}(\bm{\chi},\bm{\lambda})\nonumber\\
&=\min_{\{\bm{\chi}(\mathbf{s}_t)  \in {\cal X},~\forall \mathbf{s}_t\in \mathcal{S}\}}\,\mathbb{E}\big[{\cal
    L}(\bm{\chi}(\mathbf{s}_t),\bm{\lambda};\mathbf{s}_t)\big].
\end{align}
Note that the optimization in \eqref{eq.dual-func1} is still w.r.t. a function. To facilitate the optimization, we re-write \eqref{eq.dual-func1} relying on the so-termed \emph{interchangeability principle} \cite[Theorem 7.80]{shapiro2009}.
\begin{lemma}\label{lemm-inter}
	Let $\bm{\xi}$ denote a random variable on $\bm{\Xi}$, and ${\cal H}:=\{h(\,\cdot\,):\bm{\Xi}\rightarrow\mathbb{R}^n\}$ denote the function space of all the functions on $\bm{\Xi}$. 
	For any $\bm{\xi}\in \bm{\Xi}$, if $f(\,\cdot\,,\bm{\xi}):\mathbb{R}^n\rightarrow \mathbb{R}$ is a proper and lower semicontinuous convex function, then it follows that
	 \begin{equation}\label{eq.inter1}
	 		\min_{h(\cdot)\in{\cal H}}\,\mathbb{E}\big[f(h(\bm{\xi}),\bm{\xi})\big]=\mathbb{E}\left[\min_{\mathbf{h} \in \mathbb{R}^n}f(\mathbf{h},\bm{\xi})\right].
	 \end{equation}
\end{lemma}
Lemma \ref{lemm-inter} implies that under mild conditions, we can replace the optimization over a function space with (infinitely many) point-wise optimization problems. In the context here, we assume that $\Psi_t(\mathbf{x}_t)$ is proper, lower semicontinuous, and strongly convex (cf. Assumption \ref{assp.primal} in Section V). 
Thus, for given finite $\bm{\lambda}$ and $\mathbf{s}_t$, ${\cal L}(\,\cdot\,,\bm{\lambda};\mathbf{s}_t)$ is also strongly convex, proper and lower semicontinuous. 
Therefore, applying Lemma \ref{lemm-inter} yields 
\begin{equation}\label{eq.inter2}
\hspace{-0.1cm}\min_{\{\bm{\chi}(\cdot):{\cal S}\rightarrow{\cal X}\}}\!\!\!\mathbb{E}\big[{\cal L}(\bm{\chi}(\mathbf{s}_t),\bm{\lambda};\mathbf{s}_t)\big]\!\!=\!\mathbb{E}\big[\!\min_{\bm{\chi}(\mathbf{s}_t) \in {\cal X}}\!\!\!{\cal L}(\bm{\chi}(\mathbf{s}_t),\bm{\lambda};\mathbf{s}_t)\big]\!\!\!\!
\end{equation}
where the minimization and the expectation are interchanged.
Accordingly, we re-write \eqref{eq.dual-func1} in the following form
\begin{align}
\!\!\!{\cal D}(\bm{\lambda})\!=\!\mathbb{E}\!\left[\min_{\bm{\chi}(\mathbf{s}_t) \in {\cal
X}}\!{\cal L}(\bm{\chi}(\mathbf{s}_t),\bm{\lambda};\mathbf{s}_t)\right]\!=\!\mathbb{E}\!\left[\min_{\mathbf{x}_t \in {\cal X}}{\cal L}_t(\mathbf{x}_t,\bm{\lambda})\right]\!.\!\!
\end{align}
\end{subequations}
Likewise, for the instantaneous dual function ${\cal D}_t(\bm{\lambda})={\cal D}(\bm{\lambda};\mathbf{s}_t):=\min_{\mathbf{x}_t \in {\cal
X}}{\cal L}_t(\mathbf{x}_t,\bm{\lambda})$, the dual problem of \eqref{eq.reform2} is
\begin{align}\label{eq.dual-prob}
 \max_{\bm{\lambda}\geq \mathbf{0}} \, {\cal D}(\bm{\lambda}):=\mathbb{E}\left[{\cal
 D}_t(\bm{\lambda})\right].
\end{align}
In accordance with the ensemble primal problem \eqref{eq.reform2}, we will henceforth refer to \eqref{eq.dual-prob} as the \textit{ensemble} dual problem.

If the optimal Lagrange multiplier $\bm{\lambda}^*$ associated with \eqref{eq.reforme1} were known, then optimizing 
\eqref{eq.reform2} and consequently \eqref{eq.reform} would be equivalent to minimizing the Lagrangian
${\cal L}(\bm{\chi},\bm{\lambda}^*)$ or infinitely many instantaneous $\{{\cal L}_t(\mathbf{x}_t,\bm{\lambda}^*)\}$,
over the set ${\cal X}$ \cite{borkar2002}.
We restate this assertion as follows.

\begin{proposition}\label{prop.closedform}
Consider the optimization problem in 
\eqref{eq.reform2}. Given a realization 
$\mathbf{s}_t$, and the optimal Lagrange multiplier
$\bm{\lambda}^*$ associated with the constraints
\eqref{eq.reforme1}, the optimal instantaneous resource allocation
decision is
    \begin{equation}\label{close-up1}
        \mathbf{x}_t^*=\bm{\chi}^*(\mathbf{s}_t)\in\arg\min_{\bm{\chi}(\mathbf{s}_t)\in {\cal X}}{\cal
        L}(\mathbf{x}_t,\bm{\lambda}^*;\mathbf{s}_t)
    \end{equation}
where $\in$ accounts for possibly multiple minimizers of ${\cal
        L}_t$.
\end{proposition}
When the realizations $\{\mathbf{s}_t\}$ are obtained sequentially, one can generate a sequence of optimal solutions $\{\mathbf{x}_t^*\}$ correspondingly for the dynamic problem \eqref{eq.reform}. 
To obtain the optimal allocation in \eqref{close-up1} however, $\bm{\lambda}^*$ must be known. This fact motivates our novel ``learn-and-adapt'' stochastic dual gradient (LA-SDG)
method in Section \ref{sec.LA-SDG}. 
To this end, we will first outline the celebrated stochastic dual gradient iteration (a.k.a. Lyapunov optimization).

\subsection{Revisiting stochastic dual (sub)gradient}\label{subsec.DGD}

To solve \eqref{eq.dual-prob}, a standard gradient iteration involves sequentially taking expectations
over the distribution of $\mathbf{s}_t$ to compute the gradient. 
Note that when the Lagrangian minimization (cf. \eqref{close-up1}) admits possibly multiple minimizers, a subgradient iteration is employed instead of the gradient one \cite{bertsekas2003}. 
This is challenging because the distribution of $\mathbf{s}_t$ is typically unknown in practice.
But even if the joint probability distribution functions were
available, finding the expectations is not scalable as the dimensionality of $\mathbf{s}_t$ grows.

A common remedy to this challenge is stochastic approximation \cite{robbins1951,neely2010}, which corresponds to the following SDG iteration
\begin{subequations}\label{eq.dual-sgd}
\begin{equation}\label{eq.dual-stocg}
\bm{\lambda}_{t+1} ~=\big[\bm{\lambda}_t+ \mu \nabla{\cal
D}_t(\bm{\lambda}_t)\big]^{+},\;\forall t
\end{equation}
where $\mu$ is a positive (and typically pre-selected constant) stepsize. The
stochastic (sub)gradient $\nabla{\cal
D}_t(\bm{\lambda}_t)=\mathbf{A}\mathbf{x}_t+\mathbf{c}_t$ is an
unbiased estimate of the true (sub)gradient; that is, $\mathbb{E}[\nabla{\cal D}_t(\bm{\lambda}_t)]=\nabla{\cal D}(\bm{\lambda}_t)$.
Hence, the primal $\mathbf{x}_t$ can be found by solving
the following instantaneous sub-problems, one per $t$
\begin{equation}\label{eq.SA-sub}
    \mathbf{x}_t\in\arg\min_{\mathbf{x}_t \in {\cal X}}{\cal L}_t(\mathbf{x}_t,\bm{\lambda}_t).
\end{equation}
\end{subequations}


The iterate $\bm{\lambda}_{t+1}$ in \eqref{eq.dual-stocg} depends only on the probability distribution
of $\mathbf{s}_t$ through the stochastic (sub)gradient $\nabla{\cal
D}_t(\bm{\lambda}_t)$. Consequently, the process
$\{\bm{\lambda}_t\}$ is Markov with invariant transition
probability when $\mathbf{s}_t$ is stationary. An interesting
observation is that since $\nabla{\cal
D}_t(\bm{\lambda}_t):=\mathbf{A}\mathbf{x}_t+\mathbf{c}_t$, the
dual iteration can be written as [cf. \eqref{eq.dual-stocg}]
\begin{equation}\label{eq.dual-stocg2}
{\bm{\lambda}_{t+1}}/{\mu}=\left[{\bm{\lambda}_t}/{\mu}+
\mathbf{A}\mathbf{x}_t+\mathbf{c}_t\right]^{+},\; \forall t
\end{equation}
which coincides with \eqref{eq.probm} for
$\bm{\lambda}_t/\mu=\mathbf{q}_t$; see also \cite{neely2010,huang2011,Urg11} for a virtual queue interpretation of this parallelism.

Thanks to its low complexity and robustness to non-stationary
scenarios, SDG is widely used in
various areas, including adaptive signal processing
\cite{kong1995}, stochastic network optimization
\cite{neely2010,huang2011,huang2014}, and energy management in
power grids \cite{Urg11,sun2016}. For network management in particular, this iteration entails a cost-delay tradeoff as summarized next; see e.g.,  
\cite{neely2010}.
\begin{proposition}\label{prop.SDG}
If ${\Psi}^*$ is the optimal cost in \eqref{eq.prob} under any feasible control policy with the state distribution available, and if a constant stepsize $\mu$ is used in \eqref{eq.dual-stocg}, the SDG recursion \eqref{eq.dual-sgd} achieves an ${\cal O}(\mu)$-optimal solution in the sense that
\begin{subequations}
	\begin{equation}
    \lim_{T\rightarrow \infty} \frac{1}{T} \sum_{t=1}^{T} \mathbb{E}\left[\Psi_t\left(\mathbf{x}_t(\bm{\lambda}_t)\right)\right] \leq {\Psi}^*+{\cal O}(\mu)
\end{equation}
where $\mathbf{x}_t(\bm{\lambda}_t)$ denotes the decisions obtained from \eqref{eq.SA-sub}, and it incurs a steady-state queue length ${\cal O}(1/\mu)$, namely
\begin{equation}
    \lim_{T\rightarrow \infty} \frac{1}{T} \sum_{t=1}^{T} \mathbb{E}\left[\mathbf{q}_t\right]={\cal O}\left(\frac{1}{\mu}\right).
\end{equation}
\end{subequations}
\end{proposition}


Proposition \ref{prop.SDG} asserts that SDG with stepsize $\mu$ will asymptotically yield an ${\cal O}(\mu)$-optimal solution \cite[Prop. 8.2.11]{bertsekas2003}, and it will have steady-state queue length $\mathbf{q}_{\infty}$ inversely proportional to $\mu$. 
This optimality gap is standard, because iteration \eqref{eq.dual-stocg} with a constant stepsize\footnote{A vanishing stepsize in the stochastic approximation iterations can ensure convergence, but necessarily implies an unbounded queue length as $\mu\rightarrow 0$ \cite{neely2010}.} will converge to a neighborhood of the optimum $\bm{\lambda}^*$ \cite{kong1995}.
Under mild conditions, the optimal multiplier is bounded, i.e., $\bm{\lambda}^*={\cal O}(1)$, so that the steady-state queue length $\mathbf{q}_{\infty}$ naturally scales with ${\cal O}(1/{\mu})$ since it hovers around $\bm{\lambda}^*/\mu$; see \eqref{eq.dual-stocg2}. 
As a consequence, to achieve near optimality (sufficiently small $\mu$), SDG incurs large average queue lengths, and thus undesired average delay as per Little's law \cite{neely2010}.
To overcome this limitation, we develop next an \textit{online} approach, which can improve SDG's cost-delay tradeoff, while still preserving its affordable complexity and adaptability.

\section{Learn-and-Adapt SDG}\label{sec.LA-SDG}
Our main approach is derived in this section, by nicely leveraging both learning and optimization tools. Its decentralized implementation is also developed.

\subsection{LA-SDG as a foresighted learning scheme}\label{LA-SDG}

The intuition behind our learn-and-adapt
stochastic dual gradient (LA-SDG) approach is to incrementally
learn network state statistics from observed data while adapting
resource allocation driven by the learning process.
A key element of LA-SDG could be termed as ``foresighted'' learning
because instead of myopically learning the exact optimal
argument from empirical data, LA-SDG maintains the capability to
hedge against the risk of ``future non-stationarities.''


\begin{algorithm}[t]
\caption{LA-SDG for Stochastic Network Optimization}\label{algo1}
\begin{algorithmic}[1]
\State \textbf{Initialize:} dual iterate $\bm{\lambda}_1$,
empirical dual iterate $\hat{\bm{\lambda}}_1$, queue length
$\mathbf{q}_1$, control variable
$\bm{\theta}=\sqrt{\mu}\log^2(\mu)\cdot\mathbf{1}$, and proper
stepsizes $\mu$ and $\{\eta_t,\,\forall t\}$. \For {$t=1,2\dots$}
\State \textbf{Resource allocation (1st gradient):} \State
Construct the effective dual variable via \eqref{eq.dual-effect},
observe {\color{white}~~~~~}the current state $\mathbf{s}_t$, and obtain resource allocation $\mathbf{x}_t(\bm{\gamma}_t)$ {\color{white}~~~~}by minimizing online Lagrangian \eqref{eq.real-time1}.
\State Update the instantaneous queue length $\mathbf{q}_{t+1}$ via
\begin{equation}\label{eq.dual-queue}
    \mathbf{q}_{t+1} =\big[\mathbf{q}_t+ \big(\mathbf{A}\mathbf{x}_t(\bm{\gamma}_t)+\mathbf{c}_t\big)\big]^{+},\;\forall t.
\end{equation}
\State \textbf{Sample recourse (2nd gradient):} \State Obtain
variable $\mathbf{x}_t(\hat{\bm{\lambda}}_t)$ by solving online
Lagrangian {\color{white}~~~~}minimization with sample $\mathbf{s}_t$ via \eqref{eq.real-time2}.
\State Update the empirical dual variable $\hat{\bm{\lambda}}_{t+1}$ via \eqref{eq.dual-lambda-alg}.
\EndFor
\end{algorithmic}
\end{algorithm}

The proposed LA-SDG is summarized in Algorithm \ref{algo1}.
It involves the queue length $\mathbf{q}_t$ and
an empirical dual variable $\hat{\bm{\lambda}}_t$, along with a bias-control variable $\bm{\theta}$ to ensure that LA-SDG will attain near optimality in the steady state [cf. Theorems \ref{the.queue-stable} and
\ref{gap-onlineLA-SDG}]. At each time slot $t$, LA-SDG obtains
two stochastic gradients using the current $\mathbf{s}_t$:
One for online resource allocation, and another one for
sample learning/recourse. For the first gradient (lines 3-5), contrary to SDG that relies on the stochastic multiplier
estimate $\bm{\lambda}_t$ [cf. \eqref{eq.SA-sub}], LA-SDG minimizes the
instantaneous Lagrangian
\begin{subequations}\label{eq.res-allot}
 \begin{equation}\label{eq.real-time1}
    \mathbf{x}_t(\bm{\gamma}_t)\in\arg\min_{\mathbf{x}_t \in {\cal X}}{\cal
    L}_t(\mathbf{x}_t,\bm{\gamma}_t)
\end{equation}
which depends on what we term \textit{effective} multiplier, given by
\begin{equation}\label{eq.dual-effect}
    \!\underbrace{~~~~\bm{\gamma}_t~~~~}_{\rm effective~multiplier}=\underbrace{~~~~\hat{\bm{\lambda}}_t~~~~}_{\rm statistical~learning}+~~\underbrace{~~~\mu\mathbf{q}_t~-~\bm{\theta}~~~}_{\rm online~adaptation},\;\forall t.
\end{equation}
\end{subequations}
Variable $\bm{\gamma}_t$ also captures the effective price, which is a linear combination of the empirical $\hat{\bm{\lambda}}_t$ and the queue length
$\mathbf{q}_t$, where the control variable $\mu$ tunes the weights
of these two factors, and $\bm{\theta}$ controls the bias of $\bm{\gamma}_t$ in the steady state \cite{huang2014}. As a single pass of SDG
``wastes'' valuable online samples, LA-SDG resolves this
limitation in a learning step by evaluating a second gradient (lines 6-8); that is, LA-SDG simply finds the
stochastic gradient of \eqref{eq.dual-prob} at the previous
empirical dual variable $\hat{\bm{\lambda}}_t$, and implements a
gradient ascent update as
\begin{subequations}\label{eq.stat-learn}
	\begin{equation}\label{eq.dual-lambda-alg}
    \hat{\bm{\lambda}}_{t+1}=\big[\hat{\bm{\lambda}}_t+\eta_t \big(\mathbf{A}\mathbf{x}_t(\hat{\bm{\lambda}}_t)+\mathbf{c}_t\big)\big]^{+},\;\forall t
\end{equation}
where $\eta_t$ is a proper
diminishing stepsize, and the ``virtual'' allocation $\mathbf{x}_t(\hat{\bm{\lambda}}_t)$ can be found by solving
\begin{equation}\label{eq.real-time2}
	\mathbf{x}_t(\hat{\bm{\lambda}}_t)\in\arg\min_{\mathbf{x}_t \in {\cal X}}{\cal
L}_t(\mathbf{x}_t,\hat{\bm{\lambda}}_t).
\end{equation}
\end{subequations}
Note that different from $\mathbf{x}_t(\bm{\gamma}_t)$ in \eqref{eq.real-time1}, the ``virtual'' allocation $\mathbf{x}_t(\hat{\bm{\lambda}}_t)$ will not be physically implemented.
The multiplicative constant $\mu$ in \eqref{eq.dual-effect} controls the degree of adaptability, and allows for adaptation even in the steady state ($t\rightarrow \infty$), but the vanishing $\eta_t$ is for learning, as we shall discuss next.

The key idea of LA-SDG is to empower adaptive resource allocation (via $\bm{\gamma}_t$) with the learning process (effected through $\hat{\bm{\lambda}}_t$). As a result, the construction of $\bm{\gamma}_t$ relies on $\hat{\bm{\lambda}}_t$, but not vice versa.
For a better illustration of the effective
price \eqref{eq.dual-effect}, we call $\hat{\bm{\lambda}}_t$
the statistically learnt price to obtain the exact 
optimal argument of the expected problem \eqref{eq.dual-prob}. 
We
also call $\mu\mathbf{q}_t$ (which is exactly $\bm{\lambda}_t$ as shown in \eqref{eq.dual-stocg}) the online adaptation term since it
can track the instantaneous change of system statistics.
Intuitively, a large $\mu$ will allow the effective policy to
quickly respond to instantaneous variations so that the policy
gains improved control of queue lengths, while a small $\mu$ puts
more weight on learning from historical samples so
that the allocation strategy will incur less variance in
the steady state. In this sense, LA-SDG can attains both statistical efficiency and adaptability.

Distinctly different from SDG that combines statistical learning with resource allocation into a single adaptation step [cf. \eqref{eq.dual-stocg}], LA-SDG performs these two tasks into two intertwined steps: resource allocation \eqref{eq.res-allot}, and statistical learning \eqref{eq.stat-learn}.
The additional learning step adopts diminishing stepsize to find the ``best empirical'' dual variable from all observed network states. 
This pair of complementary gradient
steps endows LA-SDG with its attractive properties. In its transient
stage, the extra gradient evaluations and empirical
dual variables accelerate the convergence speed of SDG; while in the
steady stage, the empirical multiplier approaches the
optimal one, which significantly reduces the
steady-state queue lengths.

\begin{remark}

Readers familiar with algorithms on statistical learning and
stochastic network optimization can recognize their similarities and
differences with LA-SDG.

(P1) SDG in \cite{neely2010} involves only the first part of LA-SDG
($1$st gradient), where the allocation policy purely relies on stochastic
estimates of Lagrange multipliers or instantaneous queue lengths,
i.e., $\bm{\gamma}_t=\mu\mathbf{q}_t$. In contrast, LA-SDG
further leverages statistical learning from streaming data.

(P2) Several schemes have been developed recently for statistical
learning at scale to find $\hat{\bm{\lambda}}_t$, namely, SAG in \cite{roux2012} and
SAGA in \cite{defazio2014}. 
However, directly applying $\bm{\gamma}_t=\hat{\bm{\lambda}}_t$ to allocate resources causes infeasibility.
For a finite time $t$, $\hat{\bm{\lambda}}_t$ is $\delta$-optimal\footnote{Iterate $\hat{\bm{\lambda}}_t$ is $\delta$-optimal if $\|\hat{\bm{\lambda}}_t-\bm{\lambda}^*\|\leq{\cal O}(\delta)$, and likewise for $\delta$-feasibility.} for \eqref{eq.dual-prob}, and the primal variable $\mathbf{x}_t(\hat{\bm{\lambda}}_t)$ in turn is $\delta$-feasible with respect to \eqref{eq.reforme1} that is necessary for \eqref{eq.probn}. 
Since $\mathbf{q}_t$ essentially accumulates online constraint violations of \eqref{eq.reforme1}, it will grow linearly with $t$ and eventually become unbounded.

\end{remark}

\subsection{LA-SDG as a modified heavy-ball iteration}\label{LA-SDG:Hball}

The heavy-ball iteration belongs to the family of momentum-based
first-order methods, and has well-documented acceleration merits
in the deterministic setting \cite{polyak1987}. Motivated by its
convergence speed in solving deterministic problems, stochastic
heavy-ball methods have been also pursued recently \cite{yuan2016,liu2016}.

The stochastic version of the heavy-ball iteration is \cite{yuan2016}
\begin{align}\label{eq.dual-hball}
\bm{\lambda}_{t+1} =\bm{\lambda}_t+ \mu \nabla{\cal
D}_t(\bm{\lambda}_t)+\beta(\bm{\lambda}_t-\bm{\lambda}_{t-1}),\;\forall
t
\end{align}
where $\mu>0$ is an appropriate constant stepsize, 
$\beta\in[0,1)$ denotes the momentum factor, and the stochastic gradient $\nabla{\cal
D}_t(\bm{\lambda}_t)$ can be found by solving \eqref{eq.SA-sub} using heavy-ball iterate $\bm{\lambda}_t$.
This iteration exhibits attractive convergence rate during the initial stage, 
but its performance degrades in the steady state. Recently, the performance of momentum
iterations (heavy-ball or Nesterov) with constant stepsize
$\mu$ and momentum factor $\beta$, has been proved equivalent to SDG with
constant ${\mu}/(1-\beta)$ per iteration \cite{yuan2016}. Since SDG with a
large stepsize converges fast at the price of considerable loss in optimality,
the momentum methods naturally inherit these attributes.

To see the influence of the momentum term, consider expanding the
iteration \eqref{eq.dual-hball} as
\begin{align}\label{eq.dual-hballex}
\bm{\lambda}_{t+1}& =\bm{\lambda}_t+ \mu \nabla{\cal D}_t(\bm{\lambda}_t)+\beta(\bm{\lambda}_t-\bm{\lambda}_{t-1})\nonumber\\
&=\bm{\lambda}_t+\mu \nabla{\cal D}_t(\bm{\lambda}_t)\!+\!\beta\left[\mu \nabla{\cal D}_{t-1}(\bm{\lambda}_{t-1})\!+\!\beta(\bm{\lambda}_{t-1}\!-\!\bm{\lambda}_{t-2})\right]\nonumber\\
&=\bm{\lambda}_t+\underbrace{\mu \textstyle\sum_{\tau=1}^t\beta^{t-\tau}\nabla{\cal D}_{\tau}(\bm{\lambda}_{\tau})}_{\rm accumulated~gradient}+\underbrace{\beta^t(\bm{\lambda}_1\!-\!\bm{\lambda}_0)}_{\rm initial~state}.
\end{align}
The stochastic heavy-ball method will accelerate convergence in
the initial stage thanks to the accumulated gradients, and it will 
gradually forget the initial state. As $t$ increases however, the algorithm also incurs a
worst-case oscillation ${\cal O}({\mu}/(1-\beta))$, which degrades
performance in terms of objective values when compared to SDG with
stepsize $\mu$. This is in agreement with the theoretical analysis
in \cite[Theorem 11]{yuan2016}.

Different from standard momentum methods, LA-SDG
nicely inherits the fast convergence in the initial
stage, while reducing the oscillation of stochastic
momentum methods in the steady state. To see this, consider two consecutive iterations \eqref{eq.dual-effect}
\begin{subequations}\label{eq.dual-effect1}
    \begin{align}
\bm{\gamma}_{t+1}&=\hat{\bm{\lambda}}_{t+1}+\mu\mathbf{q}_{t+1}-\bm{\theta}\\
\bm{\gamma}_t&=\hat{\bm{\lambda}}_t+\mu\mathbf{q}_t-\bm{\theta}
\end{align}
\end{subequations}
and subtract them, to arrive at
\begin{align}\label{eq.dual-LA-SDG2}
\bm{\gamma}_{t+1} &=\bm{\gamma}_t+ \mu \left(\mathbf{q}_{t+1}-\mathbf{q}_t\right)+(\hat{\bm{\lambda}}_{t+1}-\hat{\bm{\lambda}}_t)\nonumber\\
&=\bm{\gamma}_t+ \mu \nabla{\cal D}_t(\bm{\gamma}_t)+(\hat{\bm{\lambda}}_{t+1}-\hat{\bm{\lambda}}_t),~~~\forall t.
\end{align}
Here the equalities in \eqref{eq.dual-LA-SDG2} follows from $\nabla{\cal
D}_t(\bm{\gamma}_t)=\mathbf{A}\mathbf{x}_t(\bm{\gamma}_t)+\mathbf{c}_t$ in $\mathbf{q}_t$ recursion \eqref{eq.dual-queue}, and with a sufficiently large $\bm{\theta}$, the projection in \eqref{eq.dual-queue} rarely (with sufficiently low probability) takes effect since the steady-state $\mathbf{q}_t$ will hover around $\bm{\theta}/\mu$; see the details of Theorem \ref{the.queue-stable} and the proof thereof.

Comparing the LA-SDG iteration \eqref{eq.dual-LA-SDG2} with the
stochastic heavy-ball iteration \eqref{eq.dual-hball}, both of them correct the iterates using the stochastic gradient $\nabla{\cal D}_t(\bm{\gamma}_t)$ or $\nabla{\cal D}_t(\bm{\lambda}_t)$.
However, LA-SDG incorporates the variation of a learning
sequence (also known as a reference sequence)
$\{\hat{\bm{\lambda}}_t\}$ into the recursion of the main iterate
$\bm{\gamma}_t$, other than heavy-ball's momentum term
$\beta(\bm{\lambda}_t-\bm{\lambda}_{t-1})$. Since the variation of learning iterate $\hat{\bm{\lambda}}_t$ eventually diminishes
as $t$ increases, keeping the learning sequence enables LA-SDG to
enjoy accelerated convergence in the initial (transient) stage compared to
SDG, while avoiding large oscillation in the steady state compared to
the stochastic heavy-ball method. We formally remark this obervation next.

\begin{remark}
LA-SDG offers a fresh approach to designing stochastic optimization
algorithms in a dynamic environment. While directly applying the
momentum-based iteration to a stochastic setting may lead to 
unsatisfactory steady-state performance, it is promising to carefully design a reference sequence that exactly
converges to the optimal argument. Therefore, algorithms with improved convergence (e.g., the second-order method in
\cite{zargham2013}) can also be incorporated as a reference
sequence to further enhance the performance of LA-SDG.
\end{remark}

\subsection{Complexity and distributed implementation of LA-SDG}\label{LA-SDG:dist}

This section introduces a fully distributed implementation of
LA-SDG by exploiting the problem structure of network resource
allocation. For notational brevity, collect the variables
representing outgoing links from node $i$ in
$\mathbf{x}_t^{i}:=\{x_t^{ij},\forall j\in{\cal N}_i\}$ with
${\cal N}_i$ denoting the index set of outgoing neighbors of node
$i$. Let also $\mathbf{s}_t^{i}:=[\bm{\phi}_t^{i};c_t^{i}]$ denote the random state at node $i$. It will be shown that the
learning and allocation decision per time slot $t$ is
processed locally per node $i$ based on its local state
$\mathbf{s}_t^{i}$.

To this end, rewrite the Lagrangian minimization for a general
dual variable $\bm{\lambda}\in \mathbb{R}_+^{I}$ at time $t$ as [cf.
\eqref{eq.real-time1} and \eqref{eq.real-time2}]
\begin{align}\label{eq.dist-Lam}
\min_{\mathbf{x}_t\in{\cal X}}\sum_{i\in{\cal
I}}\Psi^i(\mathbf{x}_t^i;\bm{\phi}_t^{i})+\sum_{i\in{\cal
I}}\lambda^i(\mathbf{A}_{(i,:)}\mathbf{x}_t+c_t^i)
\end{align}
where $\lambda^i$ is the $i$-th entry of vector $\bm{\lambda}$, and $\mathbf{A}_{(i,:)}$ denotes the $i$-th row of the
node-incidence matrix $\mathbf{A}$. Clearly, $\mathbf{A}_{(i,:)}$
selects entries of $\mathbf{x}_t$ associated with the in- and out-links of node $i$. Therefore, the subproblem at node $i$ is
\begin{align}\label{eq.dist-node}
\min_{\mathbf{x}_t^i\in{\cal X}^i}
\Psi^i(\mathbf{x}_t^i;\bm{\phi}_t^{i})+\sum_{j\in{\cal
N}_i}(\lambda^j-\lambda^i) x_t^{ji}
\end{align}
where ${\cal X}^i$ is the feasible set of primal variable $\mathbf{x}_t^i$.
In the case of \eqref{eq.probl}, the feasible set ${\cal X}$ can be written as a Cartesian
product of sets $\{{\cal X}^i,\forall i\}$, so that the projection of $\mathbf{x}_t$ to ${\cal X}$ is equivalent to separate projections of $\mathbf{x}_t^i$ onto ${\cal X}^i$.
Note that $\{\lambda^j,\forall j\in{\cal N}_i\}$ will be available at
node $i$ by exchanging information with the neighbors per time $t$.
Hence, given the effective multipliers
$\gamma_t^j$ ($j$-th entry of $\bm{\gamma}_t$) from its outgoing neighbors in $j\in{\cal N}_i$, node $i$ is able to form an allocation decision
$\mathbf{x}_t^i(\bm{\gamma}_t)$ by solving the convex
programs \eqref{eq.dist-node} with $\lambda^j=\gamma_t^j$; see also \eqref{eq.real-time1}.
Needless to mention, $q_t^i$ can be locally updated via \eqref{eq.dual-queue}, that is
\begin{equation}
	 q_{t+1}^i =\left[q_t^i+ \Big(\sum_{j:i\in{\cal N}_j}x^{ji}_t(\bm{\gamma}_t) -\sum_{j\in{\cal N}_i}x^{ij}_t(\bm{\gamma}_t)+c_t^i\Big)\right]^{+}
\end{equation}
where $\{x^{ji}_t(\bm{\gamma}_t)\}$ are the local measurements of arrival (departure) workloads from (to) its neighbors. 

Likewise, the tentative primal variable $\mathbf{x}_t^i(\hat{\bm{\lambda}}_t)$ can be obtained at each node locally by solving \eqref{eq.dist-node} using the current sample $\mathbf{s}_t^i$ again with $\lambda^i=\hat{\lambda}_t^i$.
By sending $\mathbf{x}_t^i(\hat{\bm{\lambda}}_t)$ to its outgoing neighbors, node $i$ can update the empirical multiplier $\hat{\lambda}_{t+1}^i$ via 
\begin{equation}
	\hat{\lambda}_{t+1}^i\!=\!\left[\hat{\lambda}_t^i\!+\!\eta_t\Big(\sum_{j:i\in{\cal N}_j}x^{ji}_t(\hat{\bm{\lambda}}_t)\! -\!\sum_{j\in{\cal N}_i}x^{ij}_t(\hat{\bm{\lambda}}_t)\!+\!c_t^i\Big)\right]^{+}
\end{equation}
which, together with the local queue length $q_{t+1}^i$, also implies that the next $\gamma_{t+1}^i$ can be obtained locally. 


Compared with the classic SDG recursion
\eqref{eq.dual-stocg}-\eqref{eq.SA-sub}, the distributed implementation of LA-SDG incurs only a factor of two increase in computational complexity. Next, we will further analytically
establish that it can improve the delay of SDG by an order of
magnitude with the same order of optimality gap.

\section{Optimality and Stability of LA-SDG}\label{sec.perf}

This section presents performance analysis of LA-SDG, which will rely on the following four assumptions.

\begin{assumption}\label{assp.iid}
The state $\mathbf{s}_t$ is bounded and i.i.d. over time $t$.
\end{assumption}
\vspace{-0.4cm}
\begin{assumption}\label{assp.primal}
$\Psi_t(\mathbf{x}_t)$ is proper, $\sigma$-strongly convex, lower semi-continuous, and has $L_{\rm p}$-Lipschitz continuous gradient.
Also, $\Psi_t(\mathbf{x}_t)$ is non-decreasing w.r.t. all entries of $\mathbf{x}_t$ over ${\cal X}$.
\end{assumption}
\vspace{-0.4cm}
\begin{assumption}\label{assp.slater}
There exists a stationary policy $\bm{\chi}(\cdot)$ satisfying
$\bm{\chi}(\mathbf{s}_t) \in {\cal X}$ for all $\mathbf{s}_t$, and
$\mathbb{E}[\mathbf{A}\bm{\chi}(\mathbf{s}_t)+\mathbf{c}_t]\leq
-\bm{\zeta}$, where $\bm{\zeta}>\mathbf{0}$ is a slack vector constant.
\end{assumption}
\vspace{-0.4cm}
\begin{assumption}\label{assp.dualgrad}
	For any time $t$, the magnitude of the constraint is bounded, that is, $\|\mathbf{A}\mathbf{x}_t+\mathbf{c}_t\|\leq M,\;\forall \mathbf{x}_t \in {\cal X}$. 
\end{assumption}


Assumption \ref{assp.iid} is typical in stochastic network resource
allocation \cite{huang2011,huang2014,eryilmaz2006}, and can be relaxed to an ergodic and stationary setting following
\cite{Ale10,duchi2012}. 
Assumption \ref{assp.primal} requires the primal objective to be well behaved, meaning that it is bounded from below and has a unique optimal solution. 
Note that non-decreasing costs with increased resources are easily guaranteed with e.g., exponential and quadratic functions in our simulations. 
In addition, Assumption 2 ensures that the dual function has favorable properties, which are important for the ensuring stability analysis.
Assumption \ref{assp.slater} is Slater's condition, which guarantees the
existence of a bounded optimal Lagrange multiplier
\cite{bertsekas2003}, and is also necessary for queue stability \cite{neely2010}.  
Assumption \ref{assp.dualgrad} guarantees boundedness of the gradient of the instantaneous dual function, which is common in performance analysis of stochastic gradient-type algorithms \cite{nemirovski2009}. 

Building upon the desirable properties of the primal problem, we next show that the corresponding dual function satisfies both smoothness and quadratic growth properties \cite{Hong2016,mark2016}, which will be critical to the subsequent analysis.

\begin{lemma}\label{lemma.QG}
Under Assumption 2, the dual function ${\cal D}(\bm{\lambda})$ in \eqref{eq.dual-prob} is $L_{\rm d}$-smooth, where $L_{\rm d}=\rho(\mathbf{A}^{\top}\mathbf{A})/\sigma$, and $\rho(\mathbf{A}^{\top}\mathbf{A})$ denotes the spectral radius of $\mathbf{A}^{\top}\mathbf{A}$. 
	In addition, if $\bm{\lambda}$ lies in a compact set, there always exists a constant $\epsilon$ such that ${\cal D}(\bm{\lambda})$ satisfies the following quadratic growth property
	\begin{equation}\label{eq.QG}
		{\cal D}(\bm{\lambda}^*)-{\cal D}(\bm{\lambda})\geq \frac{\epsilon}{2}\|\bm{\lambda}^*-\bm{\lambda}\|^2
	\end{equation}
	where $\bm{\lambda}^*$ is the optimal multiplier for the dual problem \eqref{eq.dual-prob}. 
\end{lemma}
\begin{proof}
See Appendix A.
\end{proof}

We start with the convergence of the empirical dual variables
$\hat{\bm{\lambda}}_t$. Note that the update of
$\hat{\bm{\lambda}}_t$ is a standard learning iteration from historical
data, and it is not affected by future resource allocation
decisions. Therefore, the theoretical result on SDG with
diminishing stepsize is directly applicable \cite[Sec.
2.2]{nemirovski2009}.

\begin{lemma}\label{error-emp}
Let $\hat{\bm{\lambda}}_t$ denote the empirical dual variable in
Algorithm \ref{algo1}, and $\bm{\lambda}^*$ the optimal argument for the dual
problem \eqref{eq.dual-prob}. If the stepsize
is chosen as $\eta_t=\frac{\alpha D}{M\sqrt{t}},\,\forall t$, with
a constant $\alpha>0$, a sufficient large constant $D>0$, and
$M$ as in Assumption \ref{assp.dualgrad}, then it holds
that
\begin{equation}\label{ineq.error-rec}
    \mathbb{E}\left[{\cal D}(\bm{\lambda}^*)\!-\!{\cal D}(\hat{\bm{\lambda}}_t)\right]\leq
    \max\{\alpha,\alpha^{-1}\}\,\frac{DM}{\sqrt{t}}
\end{equation}
where the expectation is over all the random states $\mathbf{s}_t$ up to $t$.
\end{lemma}
Lemma \ref{error-emp} asserts that using a diminishing stepsize,
the dual function value converges sub-linearly to the optimal
value in expectation. In principle, $D$ is the radius of the
feasible set for the dual variable $\bm{\lambda}$ \cite[Sec.
2.2]{nemirovski2009}. However, as the optimal multiplier
$\bm{\lambda}^*$ is bounded according to Assumption
\ref{assp.slater}, one can always estimate a large enough $D$, and
the estimation error will only affect the constant of the
sub-optimality bound \eqref{ineq.error-rec} through the scalar
$\alpha$. The sub-optimality bound in Lemma \ref{error-emp} holds
in expectation, which averages over all possible sample paths $\{\mathbf{s}_1,\ldots,\mathbf{s}_t\}$.

As a complement to Lemma \ref{error-emp}, the almost sure 
convergence of the empirical dual variables is established next to characterize the performance of each individual sample path.
\begin{theorem}\label{emp-dual}
For the sequence of empirical multipliers $\{\hat{\bm{\lambda}}_t\}$ in Algorithm \ref{algo1}, if the stepsizes are chosen as $\eta_t=\frac{\alpha D}{M\sqrt{t}},\forall t$, with constants $\alpha, M, D$ defined in Lemma \ref{error-emp}, it holds that
    \begin{equation}\label{eq.emp-dual}
        \lim_{t\rightarrow \infty} \hat{\bm{\lambda}}_t=\bm{\lambda}^*,\quad {\rm
        w.p.1}
    \end{equation}
    where $\bm{\lambda}^*$ is the optimal dual variable for the expected dual function minimization \eqref{eq.dual-prob}.
\end{theorem}
\begin{proof}
The proof follows the steps in \cite[Proposition 8.2.13]{bertsekas2003}, which is omitted here.
\end{proof}

Building upon the asymptotic convergence of empirical dual variables for
statistical learning, it becomes possible to analyze the online
performance of LA-SDG. Clearly, the online resource allocation
$\mathbf{x}_t$ is a function of the effective dual variable
$\bm{\gamma}_t$ and the instantaneous network state $\mathbf{s}_t$
[cf. \eqref{eq.real-time1}]. Therefore, the next step is to
show that the effective dual variable $\bm{\gamma}_t$ also
converges to the optimal argument of the expected problem
\eqref{eq.dual-prob}, which would establish that the online resource allocation
$\mathbf{x}_t$ is asymptotically optimal. However, directly
analyzing the trajectory of $\bm{\gamma}_t$ is nontrivial, because
the queue length $\{\mathbf{q}_t\}$ is coupled with the reference sequence
$\{\hat{\bm{\lambda}}_t\}$ in $\bm{\gamma}_t$. To address this issue, rewrite the
recursion of $\bm{\gamma}_t$ as
\begin{equation}\label{eq.gamma}
    \bm{\gamma}_{t+1}=\bm{\gamma}_t+(\hat{\bm{\lambda}}_{t+1}-\hat{\bm{\lambda}}_t)+\mu(\mathbf{q}_{t+1}-\mathbf{q}_t),\;\forall t
\end{equation}
where the update of $\bm{\gamma}_t$ depends on the variations of
$\hat{\bm{\lambda}}_t$ and $\mathbf{q}_t$. We will
first study the asymptotic behavior of queue lengths
$\mathbf{q}_t$, and then derive the analysis of $\bm{\gamma}_t$
using the convergence of $\hat{\bm{\lambda}}_t$ in
\eqref{eq.emp-dual}, and the recursion \eqref{eq.gamma}.

Define the time-varying target
$\tilde{\bm{\theta}}_t=\bm{\lambda}^*-\hat{\bm{\lambda}}_t+\bm{\theta}$,
which is the optimality residual of statistical learning
$\bm{\lambda}^*-\hat{\bm{\lambda}}_t$ plus the bias-control variable
$\bm{\theta}$. Per Theorem \ref{emp-dual}, it readily follows that
$\lim_{t\rightarrow
\infty}\tilde{\bm{\theta}}_t=\bm{\theta},\,{\rm w.p.1}$. By
showing that $\mathbf{q}_t$ is attracted towards
the time-varying target $\tilde{\bm{\theta}}_t/\mu$, we will further derive the stability of queue lengths.

\begin{lemma}\label{lem.drift}
With $\mathbf{q}_t$ and $\mu$ denoting queue length and stepsize, there exists a constant
$B=\Theta({1}/{\sqrt{\mu}})$, and a finite time
$T_B<\infty$, such that for all $t\geq T_B$, if
$\|\mathbf{q}_t-\tilde{\bm{\theta}}_t/\mu\|>B$, it holds in LA-SDG that
    \begin{equation}\label{eq.drift}
        \mathbb{E}\left[\left\|\mathbf{q}_{t+1}-\tilde{\bm{\theta}}_t/\mu\right\|\Big|\mathbf{q}_t\right]\leq \left\|\mathbf{q}_t-\tilde{\bm{\theta}}_t/\mu\right\|-\sqrt{\mu},\; {\rm w.p.1}.
    \end{equation}
\end{lemma}
\begin{proof}\label{pf.queue-drift}
See Appendix B.
\end{proof}
Lemma \ref{lem.drift} reveals that when $\mathbf{q}_t$ is large and deviates from the time-varying target
$\tilde{\bm{\theta}}_t/\mu$, it will be bounced back towards the
target in the next time slot. Upon establishing this drift
behavior of queues, we are on track to establish queue
stability.

\begin{theorem}\label{the.queue-stable}
With $\mathbf{q}_t,\bm{\theta}$, and $\mu$ defined in \eqref{eq.dual-effect}, there exists a constant $\tilde{B}=\!\Theta({1}/{\sqrt{\mu}})$ such that the queue length under LA-SDG converges to a neighborhood of $\bm{\theta}/\mu$ as 
\begin{subequations}
	\begin{equation}\label{eq.inflim}
\liminf_{t\rightarrow \infty}\;\; \left\|\mathbf{q}_t-\bm{\theta}/\mu\right\|\leq\tilde{B},\;\;{\rm w.p.1}.
\end{equation}
In addition, if we choose $\bm{\theta}={\cal
O}(\sqrt{\mu}\log^2(\mu))$, the long-term average expected queue length satisfies
\begin{equation}\label{eq.stt-length}
    \lim_{T\rightarrow \infty} \frac{1}{T} \sum_{t=1}^{T} \mathbb{E}\left[\mathbf{q}_t\right]={\cal O}\left(\frac{\log^2(\mu)}{\sqrt{\mu}}\right),~{\rm w.p.1}.
\end{equation}
\end{subequations}
\end{theorem}
\begin{proof}\label{pf.queue-stable}
See Appendix C.
\end{proof}
Theorem \ref{the.queue-stable} in \eqref{eq.inflim} asserts that the sequence of queue iterates converges (in the infimum sense) to a neighborhood of $\bm{\theta}/{\mu}$, where the radius of neighborhood region scales as $1/\sqrt{\mu}$.
In addition to the sample path result, \eqref{eq.stt-length} demonstrates that with a specific choice of $\bm{\theta}$, the queue length averaged over all sample paths will be ${\cal O}\left({\log^2(\mu)}/{\sqrt{\mu}}\right)$. 
Together with Theorem \ref{emp-dual}, it
suffices to have the effective dual variable converge to a
neighborhood of the optimal multiplier $\bm{\lambda}^*$; that is,
$\liminf_{t\rightarrow
\infty}\bm{\gamma}_t=\bm{\lambda}^*+\mu\mathbf{q}_t-\bm{\theta}=\bm{\lambda}^*+{\cal
O}(\sqrt{\mu}),\,{\rm w.p.1}$. 
Notice that the SDG iterate
$\bm{\lambda}_t$ in \eqref{eq.dual-stocg} will also converge to a
neighborhood of $\bm{\lambda}^*$. 
Therefore, intuitively LA-SDG
will behave similar to SDG in the steady state, and its
asymptotic performance follows from that of SDG. 
However, the difference is that through a careful choice of $\bm{\theta}$, for a sufficiently small $\mu$, LA-SDG can improve the queue length ${\cal O}\left({1}/{\mu}\right)$ under SDG by an order of magnitude.

In addition to feasibility, we formally establish in the next
theorem that LA-SDG is
asymptotically near-optimal.
\begin{theorem}\label{gap-onlineLA-SDG}
Let ${\Psi}^*$ be the optimal objective value of \eqref{eq.prob}
under any feasible policy with distribution information about the state fully available. If the
control variable is chosen as $\bm{\theta}={\cal
O}(\sqrt{\mu}\log^2(\mu))$, then with a sufficiently small $\mu$, LA-SDG yields a near-optimal solution for
\eqref{eq.prob} in the sense that
\begin{equation}\label{eq.opt-gap}
        \lim_{T\rightarrow \infty} \frac{1}{T} \sum_{t=1}^{T} \mathbb{E}\left[\Psi_t\left(\mathbf{x}_t(\bm{\gamma}_t)\right)\right] \leq {\Psi}^*+{\cal O}(\mu),\;{\rm
        w.p.1}
\end{equation}
where $\mathbf{x}_t(\bm{\gamma}_t)$ denotes the real-time
operations obtained from the Lagrangian minimization
\eqref{eq.real-time1}.
\end{theorem}

\begin{proof}
See Appendix D.
\end{proof}

Combining Theorems \ref{the.queue-stable} and
\ref{gap-onlineLA-SDG}, we are ready to state that by setting $\bm{\theta}={\cal O}(\sqrt{\mu}\log^2(\mu))$,
LA-SDG is asymptotically ${\cal O}(\mu)$-optimal with an average queue
length ${\cal O}(\log^2(\mu)/{\sqrt{\mu}})$. This result implies
that LA-SDG is able to achieve a
near-optimal cost-delay tradeoff $[\mu,\log^2(\mu)/{\sqrt{\mu}}]$; see \cite{marques12,neely2010}. 
Comparing with the standard tradeoff
$[\mu,{1}/{\mu}]$ under SDG, the learn-and-adapt design of
LA-SDG markedly improves the online performance in
terms of delay. 
Note that a better tradeoff $[\mu,\log^2(\mu)]$ has been derived in
\cite{huang2014} under the so-termed local polyhedral
assumption. 
Observe though, that the considered setting in \cite{huang2014} is different from the one here. 
While the network state set $\mathbf{\cal S}$ and the action set $\mathbf{\cal X}$ in \cite{huang2014} are discrete and countable, LA-SDG allows continuous $\mathbf{\cal S}$ and $\mathbf{\cal X}$ with possibly infinite elements, and still be amenable to efficient and scalable online operations.

\begin{figure}[t]
\centering
\vspace{-0.2cm}
\includegraphics[height=0.31\textwidth]{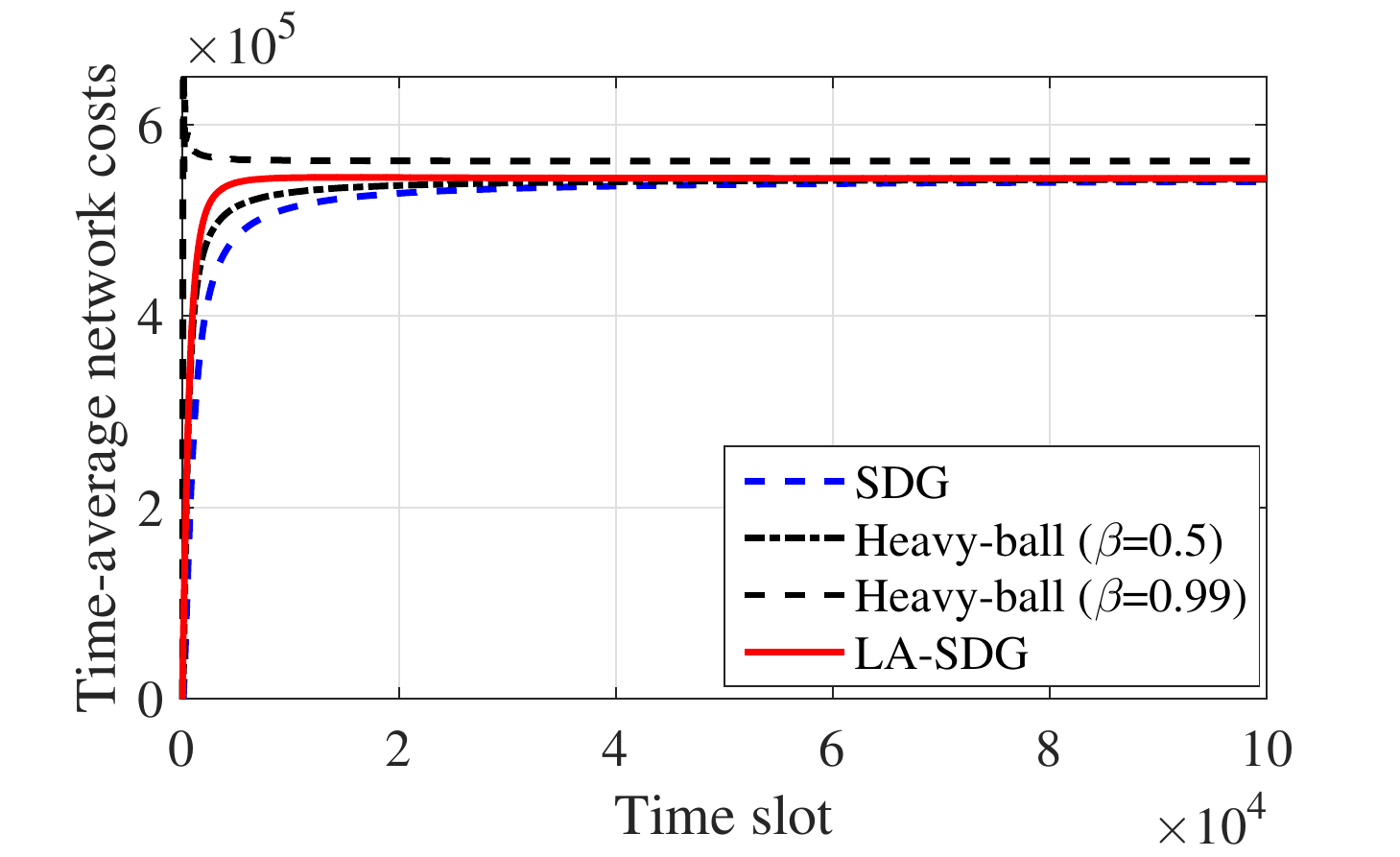}
\vspace{-0.6cm}
\caption{Comparison of time-averaged network costs.}
\label{Fig.obj}
\vspace{-0.2cm}
\end{figure}

\begin{figure}[t]
\centering
\vspace{-0.2cm}
\includegraphics[height=0.31\textwidth]{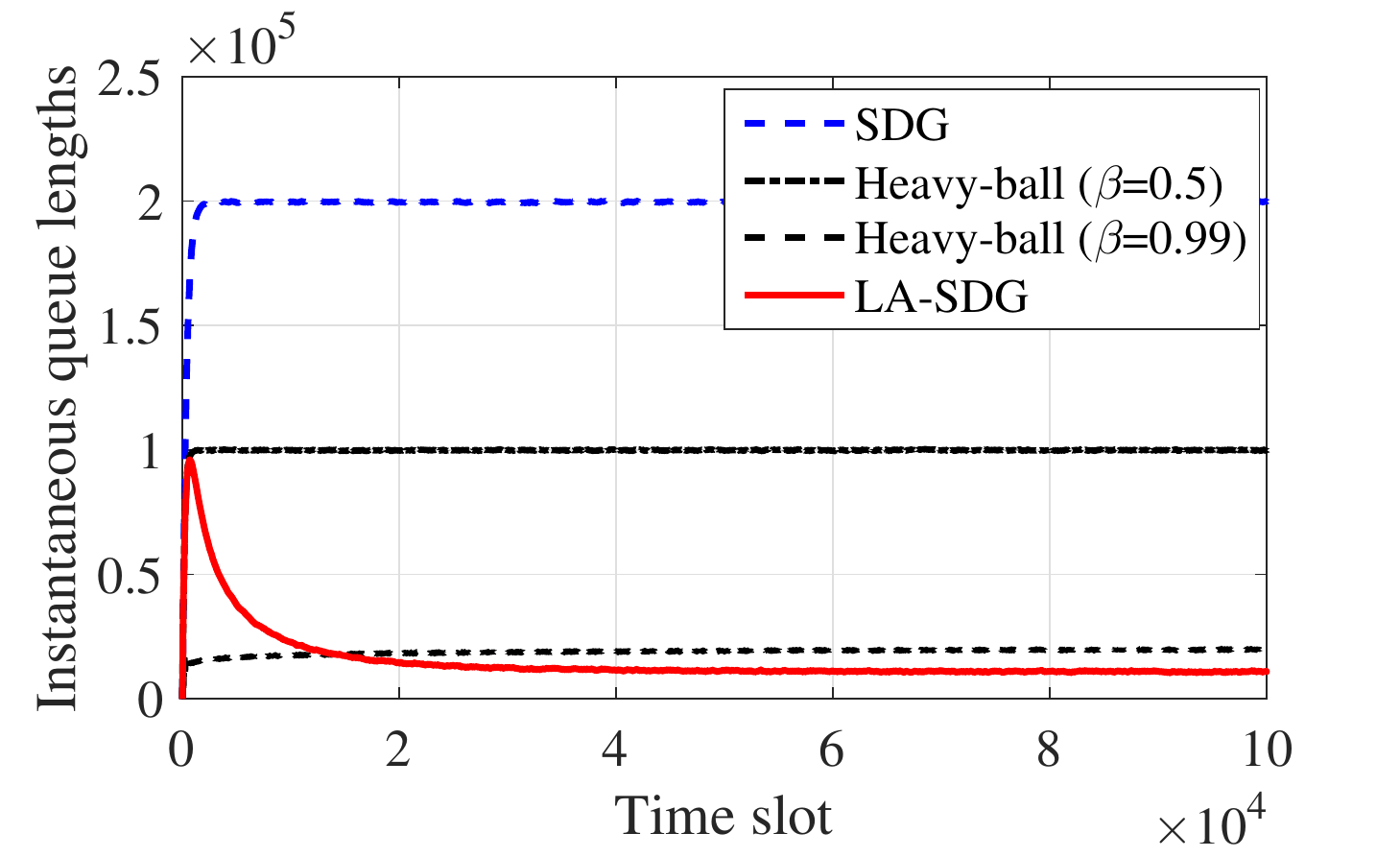}
\vspace{-0.6cm}
\caption{Instantaneous queue lengths summed over all nodes.}
\label{Fig.queue}
\end{figure}

%
%

\section{Numerical Tests}\label{sec.Num}

\begin{figure}[t]
\centering
\vspace{-0.2cm}
\includegraphics[height=0.32\textwidth]{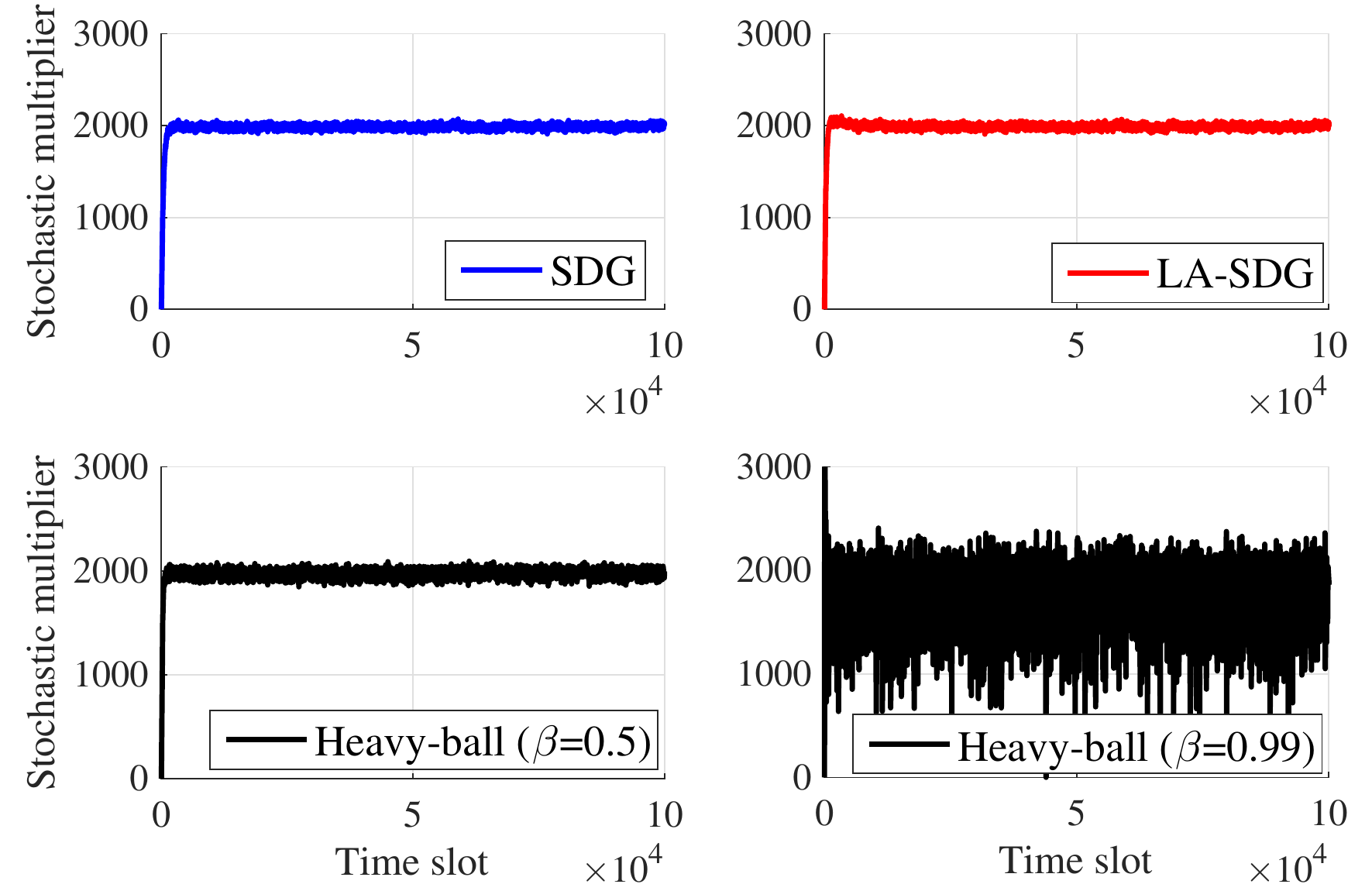}
\vspace{-0.6cm}
\caption{The evolution of stochastic multipliers at mapping node 1 ($\mu=0.2$).}
\label{Fig.dual}
\vspace{-0.2cm}
\end{figure}

This section presents numerical tests to confirm the analytical
claims and demonstrate the merits of the proposed approach. We
consider the geographical load balancing network of Section \ref{subsec.exp} with $K=10$ data centers, and $J=10$
mapping nodes. Performance is tested in terms of the time-averaged
instantaneous network cost in \eqref{eq.netcost}, namely
\begin{equation}\label{eq.simucost}
\Psi_t(\mathbf{x}_t)\!:=\sum_{k\in{\cal K}}p_t^k\left((x_t^{k0})^2
-e_t^k\right)+\sum_{j\in {\cal J}}\sum_{k\in{\cal K}}b_t^{jk}
(x_t^{jk})^2
\end{equation}
where the energy price $p_t^k$ is uniformly distributed over
$[10,30]$; samples of the renewable supply $\{e_t^k\}$ are
generated uniformly over $[10,100]$; and the
per-unit bandwidth cost is set to
$b_t^{jk}=40/\bar{x}^{jk},\forall k,j$, with bandwidth limits
$\{\bar{x}^{jk}\}$ generated from a uniform distribution
within $[100,200]$. 
The capacities at data
centers $\{\bar{x}_t^{k0}\}$ are uniformly generated from $[100,200]$. The
delay-tolerant workloads $\{c_t^j\}$ arrive at each mapping node
$j$ according to a uniform distribution over $[10,100]$.
Clearly, the cost \eqref{eq.simucost} and the state $\mathbf{s}_t$ here satisfy Assumptions 1 and 2. 
Finally, the stepsize is $\eta_t=1/\sqrt{t},\forall t$, the trade-off variable is $\mu=0.2$, and the bias correction vector is chosen as $\bm{\theta}=100\sqrt{\mu}\log^2(\mu) \bm{1}$ by default, but manually tuned in Figs. \ref{Fig.obj-all}-\ref{Fig.delay-all}.
We introduce two benchmarks: SDG in \eqref{eq.dual-stocg} (see e.g., \cite{neely2010}), and the projected stochastic
heavy-ball in \eqref{eq.dual-hball} and $\beta=0.5$ by default (see e.g.,
\cite{liu2016}). 
Unless otherwise stated, all simulated results were averaged over 50 Monte Carlo realizations.

\begin{figure}[t]
\centering
\vspace{-0.2cm}
\includegraphics[height=0.31\textwidth]{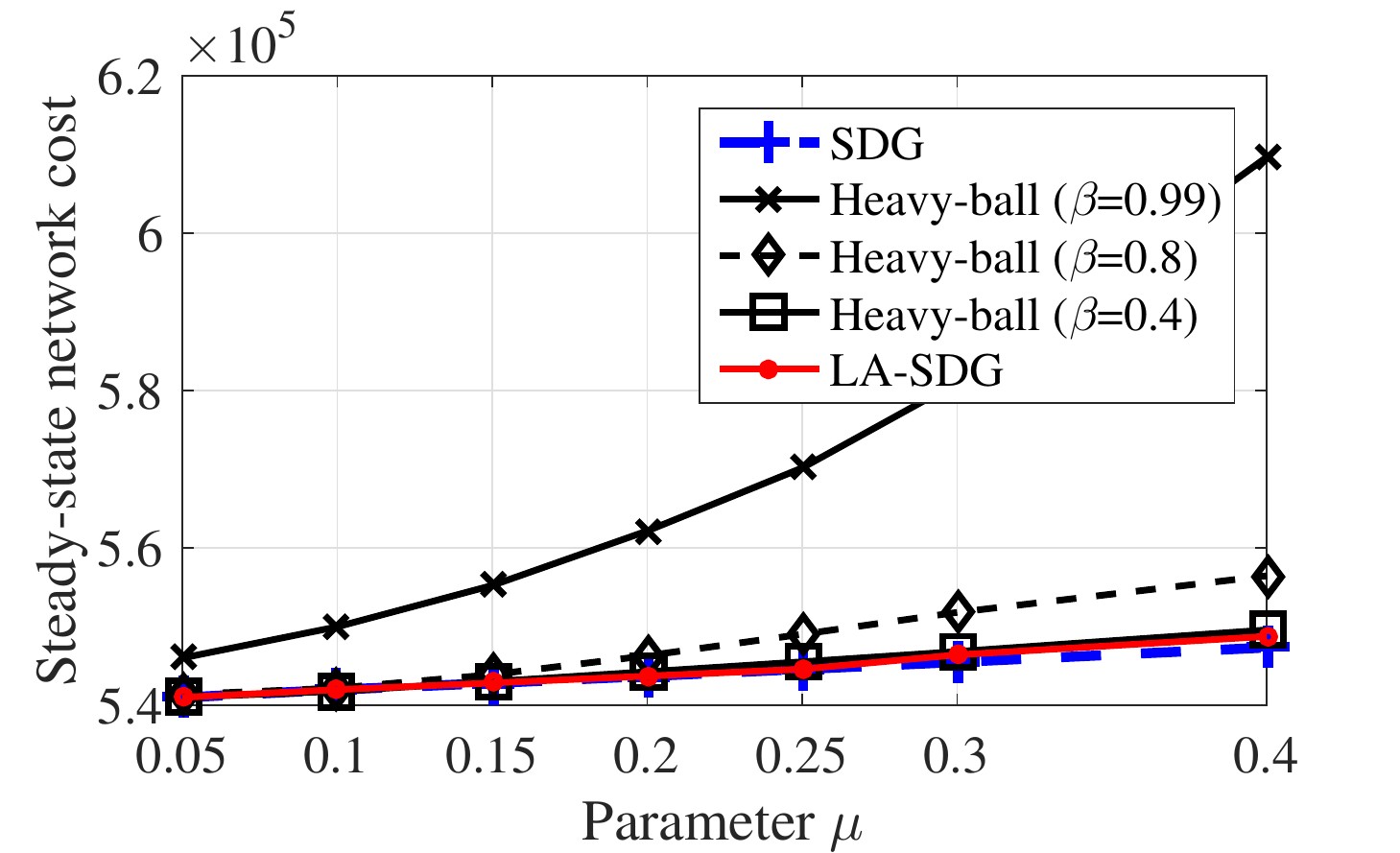}
\vspace{-0.6cm}
\caption{Comparison of steady-state network costs (after $10^6$ slots).}
\label{Fig.obj-all}
\vspace{-0.2cm}
\end{figure}

\begin{figure}[t]
\centering
\vspace{-0.2cm}
\includegraphics[height=0.31\textwidth]{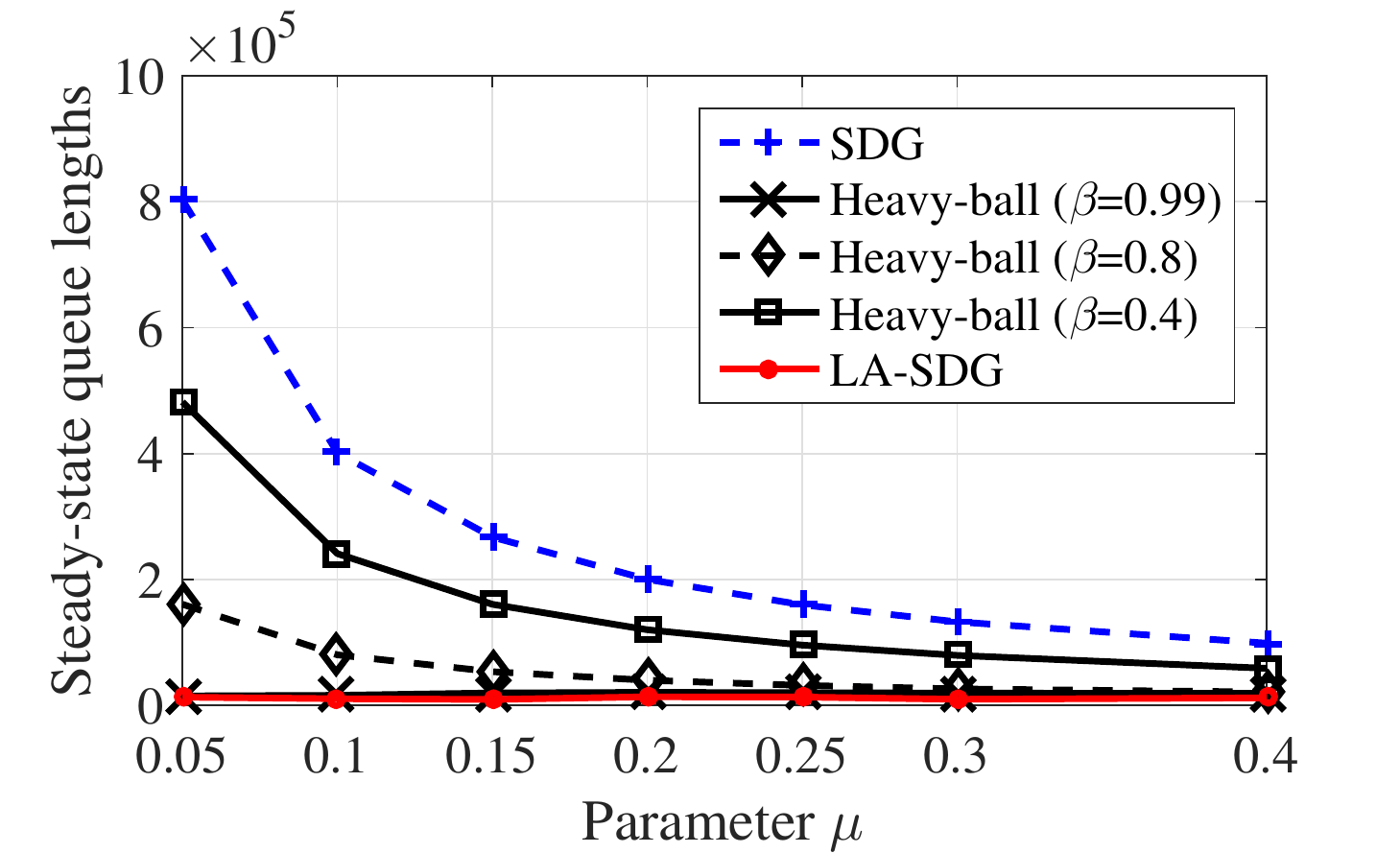}
\vspace{-0.6cm}
\caption{Steady-state queue lengths summed over all nodes (after $10^6$ slots).}
\label{Fig.delay-all}
\vspace{-0.2cm}
\end{figure}

Performance is first compared in terms of the time-averaged cost, and the instantaneous queue length in Figs. \ref{Fig.obj} and \ref{Fig.queue}. 
For the network cost, SDG, LA-SDG, and the heavy-ball iteration with $\beta=0.5$ converge to almost the same value, while the heavy-ball method with a larger momentum factor $\beta=0.99$ exhibits a pronounced optimality loss. 
LA-SDG and heavy-ball exhibit faster convergence than SDG as their running-average costs
quickly arrive at the optimal operating phase by leveraging the
learning process or the momentum acceleration. 
In this test, LA-SDG exhibits a much lower delay as its aggregated queue length is only 10\% of that for the heavy-ball method with $\beta=0.5$ and 4\% of that for SDG. 
By using a larger $\beta$, the heavy-ball method incurs a much lower queue length relative to that of SDG, but still slightly higher than that of LA-SDG. 
Clearly, our learn-and-adapt procedure improves the delay performance.


Recall that the instantaneous resource allocation can be viewed as a function of the dual variable; see Proposition \ref{prop.closedform}.
Hence, the performance differences in Figs.
\ref{Fig.obj}-\ref{Fig.queue} can be also anticipated by the different behavior of dual variables.
In Fig. \ref{Fig.dual}, the evolution of stochastic dual variables is plotted for a single Monte Carlo realization; that is the dual iterate in \eqref{eq.dual-stocg} for SDG, the momentum iteration in \eqref{eq.dual-hball} for the heavy-ball method, and the effective multiplier in \eqref{eq.dual-effect} for LA-SDG.
As illustrated in \eqref{eq.dual-hballex}, the performance of momentum
iterations is similar to SDG with larger stepsize ${\mu}/(1-\beta)$. 
This is corroborated by Fig. \ref{Fig.dual}, where the stochastic momentum iterate with $\beta=0.5$ behaves similar to the dual iterates of SDG and LA-SDG, but its oscillation becomes prohibitively high with a larger factor $\beta=0.99$, which nicely explains the higher cost in Fig. \ref{Fig.obj}.

Since the cost-delay performance is sensitive to the choice of parameters $\mu$ and $\beta$, extensive experiments are further conducted among three algorithms using different values of $\mu$ and $\beta$ in Figs.
\ref{Fig.obj-all} and \ref{Fig.delay-all}. The steady-state performance is evaluated by running algorithms for sufficiently long time, up to $10^6$ slots.
The steady-state costs of all three algorithms increase as $\mu$ becomes larger, and the costs of LA-SDG and the heavy-ball with small momentum factor $\beta=0.4$ are close to that of SDG, while the costs of the heavy-ball with larger momentum factors $\beta=0.8$ and $\beta=0.99$ are much larger than that of SDG.
Considering steady-state queue lengths (network delay), LA-SDG exhibits an order of magnitude lower amount than those of SDG and the heavy-ball with small $\beta$, under all choices of $\mu$. Note that the heavy-ball with a sufficiently large factor $\beta=0.99$ also has a very low queue length, but it incurs a higher cost than LA-SDG in Fig. \ref{Fig.obj-all} due to higher steady-state oscillation in Fig. \ref{Fig.dual}.

\section{Concluding Remarks}\label{sec.Cons}

Fast convergent resource allocation and low service delay are highly desirable attributes of stochastic network management approaches.
Leveraging recent advances in
online learning and momentum-based optimization, a novel online
approach termed LA-SDG was developed in this paper.
LA-SDG learns the network state statistics
through an additional sample recourse procedure. The associated novel
iteration can be nicely interpreted as a modified heavy-ball
recursion with an extra correction step to mitigate steady-state oscillations. It was analytically established that LA-SDG achieves a near-optimal cost-delay tradeoff
$[\mu,\log^2(\mu)/\sqrt{\mu}]$, which is better than $[\mu,1/\mu]$ of SDG, at the cost of only one extra gradient
evaluation per new datum. Our future research agenda includes novel
approaches to further hedge against non-stationarity, and improved learning schemes to
uncover other valuable statistical patterns from historical data. 

\setlength{\abovedisplayskip}{7pt}
\setlength{\belowdisplayskip}{7pt}
\appendix

Let us first state a simple but useful property regarding the primal-dual problems \eqref{eq.reform2} and \eqref{eq.dual-prob}.
\begin{proposition}\label{prop.primal-dual}
	Under Assumptions 1-3, for the constrained optimization \eqref{eq.reform2} with the optimal policy $\bm{\chi}^*(\cdot)$ and its optimal Lagrange multiplier $\bm{\lambda}^*$, it holds that $\mathbb{E}[\mathbf{A}\mathbf{x}_t^*+\mathbf{c}_t]=\mathbf{0}$ with $\mathbf{x}_t^*=\bm{\chi}^*(\mathbf{s}_t)\in{\cal X}$, and accordingly that $\nabla {\cal D}(\bm{\lambda}^*)=\mathbf{0}$.
\end{proposition}
\begin{proof}
	With $\bm{\lambda}^*$ denoting the optimal Lagrange multiplier with \eqref{eq.reforme1}, the Karush-Kuhn-Tucker (KKT) conditions \cite{bertsekas2003} are
	\begin{subequations}
		\begin{align}
		\left(\mathbb{E}[\nabla\Psi_t(\mathbf{x}_t^*)]+\mathbf{A}^{\top}\bm{\lambda}^*\right)^{\top}(\mathbb{E}[\mathbf{x}_t-\mathbf{x}_t^*])&\geq 0,\;\forall \mathbf{x}_t\in {\cal X}\label{eq.KKT1}\\
		(\bm{\lambda}^*)^{\top}\mathbb{E}[\mathbf{A}\mathbf{x}_t^*+\mathbf{c}_t]&=0\label{eq.KKT2}\\
		\mathbb{E}[\mathbf{A}\mathbf{x}_t^*+\mathbf{c}_t]\leq \mathbf{0};\;\bm{\lambda}^*&\geq \mathbf{0}\label{eq.KKT3}
	\end{align}
	\end{subequations}
	where \eqref{eq.KKT1} is the optimality condition of Lagrangian minimization, \eqref{eq.KKT2} is the complementary slackness condition, and \eqref{eq.KKT3} are the primal and dual feasibility conditions.
	
To establish the claim, let us first assume that there exists entry $k$ that the inequality constraint \eqref{eq.reforme1} is not active; i.e., $\mathbb{E}[\mathbf{A}_{(k,:)}\mathbf{x}_t^*\!+\!c_t^k]=-\zeta$ with the constant $\zeta>0$, and $\mathbf{A}_{(k,:)}$ denoting the $k$-th row of $\mathbf{A}$. As each row of $\mathbf{A}$ has at least one entry equal to $-1$, we collect all indices of entries at $k$-th row with value $-1$ in set ${\cal E}_{k}^{-1}$ so that $\mathbf{A}_{(k,e)}=-1,\forall e\in {\cal E}_{k}^{-1}$. 

Since $\mathbf{x}^*_t$ is feasible, we have $\mathbf{x}^*_t\geq\mathbf{0}$, and thus
\begin{align}
	\mathbb{E}[\mathbf{A}_{(k,:)}\mathbf{x}_t^*+c_t^k]=\mathbb{E}\left[\sum_{e\in {\cal E}}\mathbf{A}_{(k,e)}(x_t^e)^*+c_t^k\right]=-\zeta
\end{align}
which implies that
\begin{equation}
	\!\mathbb{E}\big[\textstyle\sum_{e\in {\cal E}_{k}^{-1}}(x_t^e)^*\big]\!=\!\zeta+\mathbb{E}[c_t^k\!+\!\textstyle\sum_{e\in {\cal E}\backslash{\cal E}_{k}^{-1}}\mathbf{A}_{(k,e)}(x_t^e)^*]>0.\!\!\!
\end{equation}
According \eqref{eq.KKT2}, it further follows that $(\lambda^k)^*=0$ since $(\lambda^k)^*\cdot\mathbb{E}[\mathbf{A}_{(k,:)}\mathbf{x}_t^*\!+\!c_t^k]=-(\lambda^k)^*\cdot\zeta=0$. Now we are on track to show that it contradicts with \eqref{eq.KKT1}. Since $\mathbb{E}[\sum_{e\in {\cal E}_{k}^{-1}}(x_t^e)^*]>0$, there exists at least an index $j$ such that $\mathbb{E}[(x_t^j)^*]>0,\,j \in {\cal E}_{k}^{-1}$. Choose $\mathbb{E}[\mathbf{x}_t]$ with $\mathbb{E}[x_t^{\tilde{j}}]=\mathbb{E}[(x_t^{\tilde{j}})^*],\forall {\tilde{j}}\neq j$ and $\mathbb{E}[x_t^j]=0$, to have $\mathbb{E}[\mathbf{x}_t-\mathbf{x}_t^*]=[0,\ldots,-\mathbb{E}[(x_t^j)^*],\ldots,0]^{\top}$. 
Recall that the feasible set ${\cal X}$ in \eqref{eq.probl} contains only box constraints; i.e., ${\cal X}:=\{\mathbf{x}\,|\,\mathbf{0}\leq\mathbf{x}\leq \bar{\mathbf{x}}\}$, which implies that the above selection of $\mathbf{x}_t$ is feasible.
Hence, we arrive at (with $\nabla_j\Psi_t(\mathbf{x}_t^*)$ denoting $j$-th entry of gradient)
\begin{align}\label{eq.kkt.contra}
	&\left(\mathbb{E}[\nabla\Psi_t(\mathbf{x}_t^*)]+\mathbf{A}^{\top}\bm{\lambda}^*\right)^{\top}\mathbb{E}[(\mathbf{x}_t-\mathbf{x}_t^*)]\nonumber\\
	=&-\mathbb{E}\Big[\nabla_j\Psi_t(\mathbf{x}_t^*)(x_t^j)^*-\sum_{i\in{\cal I}} (\lambda^i)^*\mathbf{A}_{(i,j)}(x_t^j)^*\Big]\nonumber\\
	\stackrel{(a)}{=}&\underbrace{-\mathbb{E}[\nabla_j\Psi_t(\mathbf{x}_t^*)(x_t^j)^*]}_{<0}-\!\!\sum_{i\in{\cal I}\backslash k} \underbrace{(\lambda^i)^*\mathbf{A}_{(i,j)}\mathbb{E}[(x_t^j)^*]}_{\geq 0}<0
\end{align}
where (a) uses $(\lambda^k)^*=0$; the first bracket follows from Assumption \ref{assp.primal} since $\nabla_j\Psi_t(\mathbf{x}_t^*)$ is monotonically increasing and $\nabla_j\Psi_t(\mathbf{x}_t^*)\geq 0$, thus for $\mathbb{E}[(x_t^j)^*]>0$ it follows $\mathbb{E}[\nabla_j\Psi_t(\mathbf{x}_t^*)]>0$; and the second bracket follows that $\bm{\lambda}^*\geq \mathbf{0}$ and each column of $\mathbf{A}$ has at most one $-1$ and $\mathbf{A}_{(k,j)}=-1$. The proof is then complete since \eqref{eq.kkt.contra} contradicts \eqref{eq.KKT1}.
\end{proof}

\subsection{Proof of Lemma \ref{lemma.QG}}\label{app.QG}
\textbf{Proof of Lipschitz continuity:}
Under Assumption \ref{assp.primal}, the primal objective $\Psi_t(\mathbf{x}_t)$ is $\sigma$-strongly convex, and the smooth constant of the dual function ${\cal D}_t(\bm{\lambda})$, or equivalently, the Lipschitz constant of gradient $\nabla{\cal D}_t(\bm{\lambda})$ directly follows from \cite[Lemma II.2]{beck2014}, which equals to $L_{\rm d}=\rho(\mathbf{A}^{\top}\mathbf{A})/\sigma$, with $\rho(\mathbf{A}^{\top}\mathbf{A})$ denoting the maximum eigenvalue of $\mathbf{A}^{\top}\mathbf{A}$. We omit the derivations of this result, and refer readers to that in \cite{beck2014}.

\textbf{Supporting lemmas for quadratic growth:}
To prove the quadratic growth property \eqref{eq.QG}, we introduce an error bound, which describes the local property of the dual function ${\cal D}(\bm{\lambda})$. 
\begin{lemma}\label{lemma.errorbd}{\cite[Lemma 2.3]{Hong2016}}
Consider the dual function in \eqref{eq.dual-func} and the feasible set ${\cal X}$ in \eqref{eq.probl} with only linear constraints. For any $\bm{\lambda}$ satisfying ${\cal D}(\bm{\lambda})> -\infty$ and $\|\nabla {\cal D}(\bm{\lambda})\|\leq \delta$, we have\begin{equation}\label{eq.errorbd}
	\|\bm{\lambda}^*-\bm{\lambda}\| \leq \xi \|\nabla {\cal D}(\bm{\lambda})\|
\end{equation}
where the scalar $\xi$ depends on the matrix $\mathbf{A}$ as well the constants $\sigma$, $L_{\rm p}$ and $L_{\rm d}$ introduced in Assumption \ref{assp.primal}.
\end{lemma}

Lemma \ref{lemma.errorbd} states a \textit{local} error bound for the dual function ${\cal D}(\bm{\lambda})$. 
The error bound is ``local'' since it holds only for $\bm{\lambda}$ close enough to the optimum $\bm{\lambda}^*$, i.e., $\|\nabla {\cal D}(\bm{\lambda})\|\leq \delta$.
Following the arguments in \cite{Hong2016} however, if the dual iterate $\bm{\lambda}$ is artificially confined to a compact set $\bm{\Lambda}$ such that $\|\bm{\lambda}\|\leq D$ with $D$ denoting the radius of $\bm{\Lambda}$,\footnote{Since the optimal multiplier is bounded per Assumption \ref{assp.slater}, one can safely find a large set $\bm{\Lambda}$ with radius $D$ to project dual iterates during optimization.} then for the case $\|\nabla {\cal D}(\bm{\lambda})\|\geq \delta$, the ratio ${\|\bm{\lambda}^*-\bm{\lambda}\|}/{\|\nabla {\cal D}(\bm{\lambda})\|}\leq D/\delta$, which implies the existence of $\xi$ satisfying \eqref{eq.errorbd} for any $\bm{\lambda}\in \bm{\Lambda}$.
Lemma \ref{lemma.errorbd} is important for establishing linear convergence rate without strong convexity \cite{Hong2016}. 
Remarkably, we will show next that this error bound is also critical to characterize the steady-state behavior of our LA-SDG scheme.

Building upon Lemma \ref{lemma.errorbd}, we next show that the ensemble dual function ${\cal D}(\bm{\lambda})$ also satisfies the so-termed Polyak-Lojasiewicz (PL) condition \cite{mark2016}.
\begin{lemma}\label{lemma.PL}
	Under Assumption \ref{assp.primal}, the local error-bound in \eqref{eq.errorbd} implies the following PL condition, namely
	\begin{equation}\label{eq.PL-ineq}
	{\cal D}(\bm{\lambda}^*)-{\cal D}(\bm{\lambda})\leq \frac{L_{\rm d}\xi^2}{2}\|\nabla {\cal D}(\bm{\lambda})\|^2
	\end{equation}
where $L_{\rm d}$ is the Lipschitz constant of the dual gradient and $\xi$ is as in \eqref{eq.errorbd}.
\end{lemma}
\begin{proof}
	Using the $L_{\rm d}$-smoothness of the dual function ${\cal D}(\bm{\lambda})$, we have for any $\bm{\lambda}$ and $\bm{\varphi}\in \mathbb{R}^I_+$ that
	\begin{align}
	\!	{\cal D}(\bm{\varphi})\leq {\cal D}(\bm{\lambda})-\langle \nabla {\cal D}(\bm{\varphi}), \bm{\lambda}-\bm{\varphi} \rangle+\frac{L_{\rm d}}{2}\|\bm{\lambda}-\bm{\varphi}\|^2\!.
	\end{align}
Choosing $\bm{\varphi}=\bm{\lambda}^*$, and using Proposition \ref{prop.primal-dual} such that $\nabla {\cal D}(\bm{\lambda}^*)=\mathbf{0}$, we have
\begin{equation}
	\!	{\cal D}(\bm{\lambda}^*)\!\leq\! {\cal D}(\bm{\lambda})+\!\frac{L_{\rm d}}{2}\|\bm{\lambda}-\bm{\lambda}^*\|^2 \!\stackrel{(a)}{\leq} {\cal D}(\bm{\lambda})\!+\!\frac{L_{\rm d}\xi^2}{2}\|\nabla {\cal D}(\bm{\lambda})\|^2\!\!\!
	\end{equation}
	where inequality (a) uses the local error-bound in \eqref{eq.errorbd}.
\end{proof}

\textbf{Proof of quadratic growth:}
The proof follows the main steps of that in \cite{mark2016}.
Building upon Lemma \ref{lemma.PL}, we next prove Lemma \ref{lemma.QG}. 
Define a function of the dual variable $\bm{\lambda}$ as $g(\bm{\lambda}):=\sqrt{{\cal D}(\bm{\lambda}^*)-{\cal D}(\bm{\lambda})}$. With the PL condition in \eqref{eq.PL-ineq}, and $\bm{\Lambda}^*$ denoting the set of optimal multipliers for \eqref{eq.dual-prob}, we have for any $\bm{\lambda}\notin \bm{\Lambda}^*$ that
\begin{equation}
\|\nabla  g(\bm{\lambda})\|^2=\frac{\|\nabla {\cal D}(\bm{\lambda})\|^2}{{\cal D}(\bm{\lambda}^*)-{\cal D}(\bm{\lambda})}\geq \frac{2}{L_{\rm d}\xi^2}
\end{equation}
which implies that $\|\nabla  g(\bm{\lambda})\|\geq \sqrt{2/(L_{\rm d}\xi^2)}$.

For any $\bm{\lambda}_0\!\notin\! \bm{\Lambda}^*$, consider the following differential equation\footnote{The time index in the proof of Lemma 1 is not related to the online optimization process, but it is useful to find the structure of the dual function.}
\begin{subequations}\label{eq.ode}
	\begin{numcases}{\hspace{-0.8cm}}
	\frac{\mathbf{d} \bm{\lambda}(\tau)}{\mathbf{d} \tau}=-\nabla  g(\bm{\lambda}(t)) \label{eq.ode1}\\
\;	\bm{\lambda}(\tau=0)=\bm{\lambda}_0      \label{eq.ode2}
\end{numcases}
\end{subequations}
which describes the continuous trajectory of $\{\bm{\lambda}(\tau)\}$ starting from $\bm{\lambda}_0$ along the direction of $-\nabla  g(\bm{\lambda}(\tau))$.
By using $\|\nabla  g(\bm{\lambda})\|\geq \sqrt{2/(L_{\rm d}\xi^2)}$, it follows that $\nabla  g(\bm{\lambda})$ is bounded below; thus, the differential equation \eqref{eq.ode} guarantees that we sufficiently reduce the value of function $g(\bm{\lambda})$, and $\bm{\lambda}(\tau)$ will eventually reach $\bm{\Lambda}^*$. 

In other words, there exists a time $T$ such that $\bm{\lambda}(T)\in \bm{\Lambda}^*$. Formally, for $\tau>T$, we have 
\begin{align}\label{eq.temp-g}
	g(\bm{\lambda}_0)&-g(\bm{\lambda}_\tau)=\int_{\bm{\lambda}_\tau}^{\bm{\lambda}_0}\left\langle \nabla g(\bm{\lambda}), \mathbf{d}\bm{\lambda}\right\rangle\nonumber\\
	&=-\int_{\bm{\lambda}_0}^{\bm{\lambda}_\tau}\left\langle \nabla g(\bm{\lambda}), \mathbf{d}\bm{\lambda}\right\rangle=-\int_{0}^{T} \left\langle \nabla g(\bm{\lambda}), \frac{\mathbf{d}\bm{\lambda}(\tau)}{\mathbf{d}\tau}\right\rangle \mathbf{d}\tau\nonumber\\
	&=\int_{0}^{T} \|\nabla g(\bm{\lambda}(\tau))\|^2 \mathbf{d}\tau\geq \int_{0}^{T} \frac{2}{L_{\rm d}\xi^2} \mathbf{d}\tau=\frac{2T}{L_{\rm d}\xi^2}.
\end{align}
Since $g(\bm{\lambda})\geq 0,\;\forall \bm{\lambda}$, we have $T\leq g(\bm{\lambda}_0)L_{\rm d}\xi^2/2$, which implies that there exists a finite time $T$ such that $\bm{\lambda}_\tau\in \bm{\Lambda}^*$.
On the other hand, the path length of trajectory $\{\bm{\lambda}(\tau)\}$ will be longer than the projection distance between $\bm{\lambda_0}$ and the closest point in $\bm{\Lambda}^*$ denoted as $\bm{\lambda}^*$, that is,
\begin{equation}\label{eq.temp-int}
	\int_{0}^{T} \left\| \frac{\mathbf{d}\bm{\lambda}(\tau)}{\mathbf{d}\tau}\right\| \mathbf{d}\tau=\int_{0}^{T} \left\|\nabla g(\bm{\lambda}(\tau))\right\| \mathbf{d}\tau\geq \|\bm{\lambda_0}-\bm{\lambda}^*\|
\end{equation}
and thus we have from \eqref{eq.temp-g} that 
\begin{align}
	g(\bm{\lambda}_0)&-g(\bm{\lambda}_\tau)=\int_{0}^{T} \|\nabla g(\bm{\lambda}(\tau))\|^2 \mathbf{d}\tau\\
	&\geq \int_{0}^{T} \|\nabla g(\bm{\lambda}(\tau))\|  \sqrt{\frac{2}{L_{\rm d}\xi^2}} \mathbf{d}\tau\stackrel{(b)}{\geq}  \sqrt{\frac{2}{L_{\rm d}\xi^2}}\|\bm{\lambda_0}-\bm{\lambda}^*\|\nonumber
\end{align}
where (b) follows from \eqref{eq.temp-int}. Choosing $T$ such that $g(\bm{\lambda}_T)=0$, we have 
\begin{equation}\label{eq.gg48}
	g(\bm{\lambda}_0)\geq \sqrt{\frac{2}{L_{\rm d}\xi^2}}\|\bm{\lambda_0}-\bm{\lambda}^*\|.
\end{equation}
Squaring both sides of \eqref{eq.gg48}, the proof is complete, since $\epsilon$ is defined as $\epsilon:=2/(L_{\rm d}\xi^2)$ and  $\bm{\lambda_0}$ can be any point outside the set of optimal multipliers.

\subsection{Proof of Lemma \ref{lem.drift}}\label{app.B}

Since $\hat{\bm{\lambda}}_t$ converges to $\bm{\lambda}^*,\; {\rm w.p.1}$ according to Theorem \ref{emp-dual}, there exists a finite time $T_{\theta}$ such that for $t>T_{\theta}$, we have $\|\bm{\lambda}^*-\hat{\bm{\lambda}}_t\|\leq \|\bm{\theta}\|$. In such case, it follows that $\tilde{\bm{\theta}}_t=\bm{\lambda}^*-\hat{\bm{\lambda}}_t+\bm{\theta}\geq \mathbf{0}$, since $\bm{\theta}\geq \mathbf{0}$. Therefore, we have
\begin{align}\label{ineq.drift0}
	\|\mathbf{q}_{t+1}-&\tilde{\bm{\theta}}_t/\mu\|^2=\|[\mathbf{q}_t+\mathbf{A}\mathbf{x}_t+\mathbf{c}_t]^+-[\tilde{\bm{\theta}}_t/\mu]^+\|^2\\
	\stackrel{(a)}{\leq}&\|\mathbf{q}_t+\mathbf{A}\mathbf{x}_t+\mathbf{c}_t-\tilde{\bm{\theta}}_t/\mu\|^2\nonumber\\
	\stackrel{(b)}{\leq} & \|\mathbf{q}_t-\tilde{\bm{\theta}}_t/\mu\|^2+2(\mathbf{q}_t-\tilde{\bm{\theta}}_t/\mu)^{\top}(\mathbf{A}\mathbf{x}_t+\mathbf{c}_t)+M^2\nonumber
	\end{align}
where (a) comes from the non-expansive property of the projection, and (b) is due to the bound $M$ in Assumption \ref{assp.dualgrad}.

The RHS of \eqref{ineq.drift0} can be upper bounded by
\begin{align}\label{ineq.drift1}
	& \|\mathbf{q}_t-\tilde{\bm{\theta}}_t/\mu\|^2+2(\mathbf{q}_t-\tilde{\bm{\theta}}_t/\mu)^{\top}(\mathbf{A}\mathbf{x}_t+\mathbf{c}_t)+M^2\nonumber\\
	\stackrel{(c)}{=} & \|\mathbf{q}_t-\tilde{\bm{\theta}}_t/\mu\|^2+\frac{2}{\mu}\left(\bm{\gamma}_t-\bm{\lambda}^*\right)^{\top}(\mathbf{A}\mathbf{x}_t+\mathbf{c}_t)+M^2
\end{align}
where (c) uses the definitions $\tilde{\bm{\theta}}_t:=\bm{\lambda}^*-\hat{\bm{\lambda}}_t+\bm{\theta}$, and $\bm{\gamma}_t:=\hat{\bm{\lambda}}_t+\mu\mathbf{q}_t-\bm{\theta}$. 
Since $\mathbf{A}\mathbf{x}_t+\mathbf{c}_t$ is the stochastic subgradient of the concave function ${\cal D}(\bm{\lambda})$ at $\bm{\lambda}=\bm{\gamma}_t$ [cf. \eqref{eq.real-time1}], we have
\begin{equation}\label{eq.lemm3-subg}
	\mathbb{E}\left[\left(\bm{\gamma}_t-\bm{\lambda}^*\right)^{\top}(\mathbf{A}\mathbf{x}_t+\mathbf{c}_t)\right]\leq {\cal D}(\bm{\gamma}_t)-{\cal D}(\bm{\lambda}^*).
\end{equation}

Taking expectations on \eqref{ineq.drift0}-\eqref{ineq.drift1} over the random state $\mathbf{s}_t$ conditioned on $\mathbf{q}_t$ and using \eqref{eq.lemm3-subg}, we arrive at 
\begin{equation}\label{ineq.50}
	\mathbb{E}\left[\|\mathbf{q}_{t+1}\!-\!\tilde{\bm{\theta}}_t/\mu\|^2\right]\!\leq\!\|\mathbf{q}_t-\tilde{\bm{\theta}}_t/\mu\|^2\!+\frac{2}{\mu}\left({\cal D}(\bm{\gamma}_t)-{\cal D}(\bm{\lambda}^*)\right)+M^2
\end{equation}
where we use the fact that ${\cal D}(\bm{\lambda}):=\mathbb{E}\left[{\cal D}_t(\bm{\lambda})\right]$ in \eqref{eq.dual-prob}. 
Using the quadratic growth property of ${\cal D}(\bm{\lambda})$ in \eqref{eq.QG} of Lemma \ref{lemma.QG}, the recursion \eqref{ineq.50} further leads to 
\begin{align}\label{ineq.drift}
	\mathbb{E}\big[\|\mathbf{q}_{t+1}-& \tilde{\bm{\theta}}_t/\mu\|^2\big] \leq\|\mathbf{q}_t-\tilde{\bm{\theta}}_t/\mu\|^2\!-\frac{2\epsilon}{\mu}\|\bm{\gamma}_t-\bm{\lambda}^*\|^2+M^2\nonumber\\
	\stackrel{(d)}{=} &\|\mathbf{q}_t-\tilde{\bm{\theta}}_t/\mu\|^2-2\mu\epsilon\|\mathbf{q}_t-\tilde{\bm{\theta}}_t/\mu\|^2+M^2
\end{align}
where equality (d) uses the definitions $\tilde{\bm{\theta}}_t:=\bm{\lambda}^*-\hat{\bm{\lambda}}_t+\bm{\theta}$ and $\bm{\gamma}_t:=\hat{\bm{\lambda}}_t+\mu\mathbf{q}_t-\bm{\theta}$, implying that $\bm{\gamma}_t-\bm{\lambda}^*=\mu\mathbf{q}_t-\tilde{\bm{\theta}}_t$.

Now considering (cf. \eqref{ineq.drift})
\begin{equation}\label{ineq.vieta}
	-2\mu\epsilon\|\mathbf{q}_t-\tilde{\bm{\theta}}_t/\mu\|^2+M^2\leq -2\sqrt{\mu}\|\mathbf{q}_t-\tilde{\bm{\theta}}_t/\mu\|+\mu
\end{equation}
and plugging it back into \eqref{ineq.drift} yields 
\begin{align}\label{ineq.drift3}
	\mathbb{E}\big[\|\mathbf{q}_{t+1}- \tilde{\bm{\theta}}_t/\mu\|^2\big] &\leq \|\mathbf{q}_t-\tilde{\bm{\theta}}_t/\mu\|^2-2\sqrt{\mu}\|\mathbf{q}_t-\tilde{\bm{\theta}}_t/\mu\|+\mu\nonumber\\
	&=\big(\|\mathbf{q}_t-\tilde{\bm{\theta}}_t/\mu\|-\sqrt{\mu}\big)^2.
\end{align}
By the convexity of $(\,\cdot\,)^2$, we further arrive at
\begin{align}
	\mathbb{E}\left[\|\mathbf{q}_{t+1}-\tilde{\bm{\theta}}_t/\mu\|\right]^2&\leq \mathbb{E}\big[\|\mathbf{q}_{t+1}- \tilde{\bm{\theta}}_t/\mu\|^2\big]\nonumber\\
	&\leq\left(\|\mathbf{q}_t-\tilde{\bm{\theta}}_t/\mu\|-\sqrt{\mu}\right)^2
\end{align}
which directly implies the argument \eqref{eq.drift} in the lemma.
By checking Vieta's formulas for second-order equations, there exists $B=\Theta(\frac{1}{\sqrt{\mu}})$ such that for $\|\mathbf{q}_t-\tilde{\bm{\theta}}_t/\mu\|>B$, inequality \eqref{ineq.vieta} holds, and thus the lemma follows readily.

\subsection{Proof of Theorem \ref{the.queue-stable}}\label{app.C}
\textbf{Proof of \eqref{eq.inflim} in Theorem \ref{the.queue-stable}:}
Theorem \ref{emp-dual} asserts that $\hat{\bm{\lambda}}_t$ eventually converges to the optimum $\bm{\lambda}^*,\,{\rm w.p.1}$. 
Hence, there always exists a finite time $T_{\rho}$ and an arbitrarily small $\rho$ such that for $t>T_{\rho}$, it holds that $\|\bm{\lambda}^*/\mu-\hat{\bm{\lambda}}_t/\mu\|\leq \rho$. Using the definition 
$\tilde{\bm{\theta}}_t=\bm{\lambda}^*-\hat{\bm{\lambda}}_t+\bm{\theta}$, it then follows by the triangle inequality that 
\begin{equation}\label{eq.the1-3}
	\left|\|\mathbf{q}_t-\tilde{\bm{\theta}}_t/\mu\|-\|\mathbf{q}_t-\bm{\theta}/\mu\|\right|\leq \|\bm{\lambda}^*/\mu-\hat{\bm{\lambda}}_t/\mu\|\leq\rho
\end{equation}
which also holds for $\mathbf{q}_{t+1}$.

Using \eqref{eq.the1-3} and the conditional drift \eqref{eq.drift} in Lemma \ref{lem.drift}, for $t\!>\!T_{\rho}$ and $\|\mathbf{q}_t-\tilde{\bm{\theta}}_t/\mu\|\!>\!B\!=\!\Theta({1}/{\sqrt{\mu}})$, it holds that
\begin{align}\label{eq.rho}
	 &\mathbb{E}\left[\left\|\mathbf{q}_{t+1}-\bm{\theta}/\mu\right\|\Big|\mathbf{q}_t\right]\leq  \mathbb{E}\left[\left\|\mathbf{q}_{t+1}-\tilde{\bm{\theta}}_t/\mu\right\|\Big|\mathbf{q}_t\right]+\rho\nonumber\\
	 \leq &\left\|\mathbf{q}_t-\tilde{\bm{\theta}}_t/\mu\right\|-\sqrt{\mu} +\rho\leq \left\|\mathbf{q}_t-\bm{\theta}/\mu\right\|-\sqrt{\mu} +2\rho.
\end{align}
Choosing $\rho$ such that $\sqrt{\tilde{\mu}}:=\sqrt{\mu}-2\rho< 0$, then for $t\!>\!T_{\rho}$ and $\|\mathbf{q}_t-\bm{\theta}/\mu\|\!>\tilde{B}\!:=\!B+\rho\!=\!\Theta({1}/{\sqrt{\mu}})$, we have 
\begin{equation}\label{eq.the1-4}
	 \mathbb{E}\left[\left\|\mathbf{q}_{t+1}-\bm{\theta}/\mu\right\|\Big|\mathbf{q}_t\right]\leq \left\|\mathbf{q}_t-\bm{\theta}/\mu\right\|-\sqrt{\tilde{\mu}}.
\end{equation}

Leveraging \eqref{eq.the1-4}, we first show \eqref{eq.inflim} by constructing a super-martingale.
Define the stochastic process $a_t$ as
\begin{subequations}\label{eq.supma}
	\begin{equation}\label{eq.supma1}
	a_t:=\left\|\mathbf{q}_t-\bm{\theta}/\mu\right\|\cdot \mathds{1}\left\{\min_{\tau\leq t}\left\|\mathbf{q}_{\tau}-\bm{\theta}/\mu\right\|>\tilde{B}\right\},\,\forall t
\end{equation}
and likewise the stochastic process $b_t$ as
\begin{equation}\label{eq.supma2}
	b_t:=\sqrt{\tilde{\mu}}\cdot \mathds{1}\left\{\min_{\tau\leq t}\left\|\mathbf{q}_{\tau}-\bm{\theta}/\mu\right\|>\tilde{B}\right\},\;\forall t.
\end{equation}
\end{subequations}
Clearly, $a_t$ tracks the distance between $\mathbf{q}_t$ and $\bm{\theta}/\mu$ until the distance becomes smaller than $\tilde{B}$ for the first time; and $b_t$ stops until $\left\|\mathbf{q}_t-\bm{\theta}/\mu\right\|\leq \tilde{B}$ for the first time as well.

With the definitions of $a_t$ and $b_t$, one can easily show that the recursion \eqref{eq.the1-4} implies
\begin{equation}\label{eq.supma3}
	 \mathbb{E}\left[a_{t+1}|\mathbf{\cal F}_t\right]\leq a_t-b_t
\end{equation}
where $\mathbf{\cal F}_t$ is the so-termed sigma algebra measuring the history of two processes.
As $a_t$ and $b_t$ are both nonnegative, \eqref{eq.supma3} allows us to apply the super-martingale convergence theorem \cite[Theorem E7.4]{kong1995}, which almost surely establishes that: (i) the sequence $a_t$ converges to a limit; and (ii) the summation $\sum_{t=1}^{\infty}b_t<\infty$. Note that (ii) implies that $\lim_{t\rightarrow \infty}b_t=0$, ${\rm w.p.1}$.
Since $\sqrt{\tilde{\mu}}>0$, it follows that the indicator function of $b_t$ eventually becomes null and thus
\begin{equation}\label{eq.limqt}
\liminf_{t\rightarrow \infty}\;\; \left\|\mathbf{q}_t-\bm{\theta}/\mu\right\|\leq\tilde{B},\;\;{\rm w.p.1}
\end{equation}
which establishes that $\mathbf{q}_t$ will eventually visit and then hover around a neighborhood of the reference point $\bm{\theta}/\mu$.

\textbf{Proof of \eqref{eq.stt-length} in Theorem \ref{the.queue-stable}:}
In complement to the sample-path result in \eqref{eq.limqt}, we next derive \eqref{eq.stt-length}, which captures the long-term queue lengths averaged over all sample paths.

Similar to \eqref{ineq.drift0}, we have
\begin{align}\label{eq.the2-69}
	&\|\mathbf{q}_{t+1}-\bm{\lambda}^*/\mu\|^2 \leq\\
	&\qquad\qquad \|\mathbf{q}_t-\bm{\lambda}^*/\mu\|^2+2(\mathbf{q}_t-\bm{\lambda}^*/\mu)^{\top}(\mathbf{A}\mathbf{x}_t+\mathbf{c}_t)+M^2. \nonumber
	\end{align}
Using the definition $\bm{\gamma}_t:=\hat{\bm{\lambda}}_t+\mu\mathbf{q}_t-\bm{\theta}$, \eqref{eq.the2-69} can be written as
\begin{align}\label{eq.the2-69}
	&\|\mathbf{q}_{t+1}-\bm{\lambda}^*/\mu\|^2 \leq \|\mathbf{q}_t-\bm{\lambda}^*/\mu\|^2\\
	&+\frac{2}{\mu}(\bm{\gamma}_t-\bm{\lambda}^*)^{\top}(\mathbf{A}\mathbf{x}_t+\mathbf{c}_t)+\frac{2}{\mu}(\bm{\theta}-\hat{\bm{\lambda}}_t)^{\top}(\mathbf{A}\mathbf{x}_t+\mathbf{c}_t)+M^2. \nonumber
	\end{align}

Defining the Lyapunov drift as $\Delta(\mathbf{q}_t)\!:=\!\frac{1}{2}(\|\mathbf{q}_{t+1}-\bm{\lambda}^*/\mu\|^2\!-\!\|\mathbf{q}_t-\bm{\lambda}^*/\mu\|^2)$ and taking expectations on \eqref{eq.the2-69} over $\mathbf{s}_t$ conditioned on $\mathbf{q}_t$, we have
\begin{align}\label{eq.the2-70}
	\!\!&\mu\mathbb{E}\left[\Delta(\mathbf{q}_t)\right]\leq \mathbb{E}\Big[(\bm{\gamma}_t-\bm{\lambda}^*)^{\!\top}\!(\mathbf{A}\mathbf{x}_t\!+\!\mathbf{c}_t)\Big]\nonumber\\
	&\qquad\qquad\qquad\qquad~+\mathbb{E}\Big[(\bm{\theta}\!-\!\hat{\bm{\lambda}}_t)^{\!\top}\!(\mathbf{A}\mathbf{x}_t+\mathbf{c}_t)\Big]\!+\!{\mu M^2}/{2}\nonumber\\
	&\stackrel{(b)}{\leq} {\cal D}(\bm{\gamma}_t)\!-\!{\cal D}(\bm{\lambda}^*)\!+\mathbb{E}\Big[(\bm{\theta}\!-\!\hat{\bm{\lambda}}_t)^{\!\top}\!(\mathbf{A}\mathbf{x}_t\!+\!\mathbf{c}_t)\Big]\!+\!{\mu M^2}/{2}
	\end{align}
	where (b) follows from \eqref{eq.lemm3-subg}.

Summing both sides over $t=1,\ldots,T$, taking expectations over all possible $\mathbf{q}_t$, and dividing both sides by $T$, we arrive at
\begin{align}\label{eq.the2-71}
	&\frac{\mu}{2T}\left(\mathbb{E}\left[\|\mathbf{q}_{T+1}-\bm{\lambda}^*/\mu\|^2\right]-\mathbb{E}\left[\|\mathbf{q}_{1}-\bm{\lambda}^*/\mu\|^2\right]\right)\leq\\
	&\frac{1}{T}\!\sum_{t=1}^T\mathbb{E}[{\cal D}(\bm{\gamma}_t)]\!-\!{\cal D}(\bm{\lambda}^*)\!+\!\frac{1}{T}\!\sum_{t=1}^T\mathbb{E}\Big[(\bm{\theta}\!-\!\hat{\bm{\lambda}}_t)^{\!\top}\!(\mathbf{A}\mathbf{x}_t\!+\!\mathbf{c}_t)\Big]\!\!+\!\frac{\mu M^2}{2}.\nonumber
\end{align}

First, it is easy to show that
\begin{align}\label{eq.the2-72}
	&\lim_{T\rightarrow\infty} \frac{\mu}{2T}\left(\mathbb{E}\left[\left\|\mathbf{q}_{T+1}-{\bm{\lambda}^*}/{\mu}\right\|^2\right]-\mathbb{E}\left[\left\|\mathbf{q}_{1}-{\bm{\lambda}^*}/{\mu}\right\|^2\right]\right)\nonumber\\
	\stackrel{(c)}{\geq}& - \lim_{T\rightarrow\infty} \frac{\mu}{2T}\mathbb{E}\left[\left\|\mathbf{q}_{1}-{\bm{\lambda}^*}/{\mu}\right\|^2\right]\stackrel{(d)}{=}0
\end{align}
where (c) holds since $\left\|\mathbf{q}_{T+1}-{\bm{\lambda}^*}/{\mu}\right\|^2\geq 0$, and (d) follows from the boundedness of $\left\|\mathbf{q}_{1}-{\bm{\lambda}^*}/{\mu}\right\|^2$.

We next argue that the following equality holds  
\begin{equation}\label{eq.the2-74}
	\lim_{T\rightarrow \infty} ({1}/{T})\; \textstyle \sum_{t=1}^{T}\mathbb{E}\left[(\bm{\theta}-\hat{\bm{\lambda}}_t)^{\top}(\mathbf{A}\mathbf{x}_t+\mathbf{c}_t)\right]={\cal O}(\mu). 
\end{equation}
Rearranging terms in \eqref{eq.the2-74} leads to
\begin{align}\label{eq.the2-43}
	&\lim_{T\rightarrow \infty} \frac{1}{T} \sum_{t=1}^{T}\mathbb{E}\left[(\bm{\theta}-\hat{\bm{\lambda}}_t)^{\top}(\mathbf{A}\mathbf{x}_t+\mathbf{c}_t)\right]\\
=&\lim_{T\rightarrow \infty} 	\frac{1}{T} \sum_{t=1}^{T}\mathbb{E}\left[\left(\bm{\theta}-\bm{\lambda}^*+(\bm{\lambda}^*\!-\!\bm{\theta}\!-\!\hat{\bm{\lambda}}_t\!+\!\bm{\theta})\right)^{\top}(\mathbf{A}\mathbf{x}_t+\mathbf{c}_t)\right].\nonumber
\end{align}
Since $\hat{\bm{\lambda}}_t$ converges to $\bm{\lambda}^*\, {\rm w.p.1}$ according to Theorem \ref{emp-dual}, there always exists a finite time $T_{\rho}$ such that for $t>T_{\rho}$, which implies that $\|\bm{\lambda}^*-\bm{\theta}-(\hat{\bm{\lambda}}_t-\bm{\theta})\|\leq \rho,\;{\rm w.p.1}$. Hence, together with the Cauchy-Schwarz inequality, we have 
\begin{align}\label{eq.the2-74-2}
	&(\bm{\lambda}^*\!-\!\bm{\theta}\!-\!\hat{\bm{\lambda}}_t\!+\!\bm{\theta})^{\top}(\mathbf{A}\mathbf{x}_t+\mathbf{c}_t)\nonumber\\
\leq &\|\bm{\lambda}^*\!-\!\bm{\theta}\!-\!(\hat{\bm{\lambda}}_t\!-\!\bm{\theta})\| \|\mathbf{A}\mathbf{x}_t\!+\!\mathbf{c}_t\|\stackrel{(d)}{\leq} \rho M={\cal O}(\rho)
\end{align}
where (d) follows since $T_{\rho}\!<\!\infty$ and constant $M$ is as in Assumption \ref{assp.dualgrad}.
Plugging \eqref{eq.the2-74-2} into \eqref{eq.the2-43}, it follows that
\begin{align}\label{eq.the2-rho74}
&\lim_{T\rightarrow \infty} ({1}/{T}) \textstyle \sum_{t=1}^{T} \mathbb{E}\left[(\bm{\theta}-\hat{\bm{\lambda}}_t)^{\top}(\mathbf{A}\mathbf{x}_t+\mathbf{c}_t)\right]\\
\leq&\lim_{T\rightarrow \infty} 	({1}/{T})\; \textstyle \sum_{t=1}^{T} \mathbb{E}\left[(\bm{\theta}-\bm{\lambda}^*)^{\top}(\mathbf{A}\mathbf{x}_t+\mathbf{c}_t)\right]+{\cal O}(\rho)\nonumber\\
\stackrel{(e)}{\leq}&\|\bm{\lambda}^*-\bm{\theta}\|\cdot \left\|\lim_{T\rightarrow \infty} 	({1}/{T})\; \textstyle \sum_{t=1}^{T} \mathbb{E}\left[-\mathbf{A}\mathbf{x}_t-\mathbf{c}_t\right]\right\|+{\cal O}(\rho)\nonumber
\end{align}
where (e) simply follows from the Cauchy-Schwarz inequality. 

Building upon \eqref{eq.the1-4}, one can follow the arguments in \cite[Theorem 4]{huang2011} to show that there exist constants $D_1\!=\!\Theta(1/\mu)$, and $D_2\!=\!\Theta(\sqrt{\mu})$, for any $d$, to obtain a large deviation bound as
\begin{equation}\label{eq.large-dev}
	\lim_{T\rightarrow \infty} \frac{1}{T} \sum_{t=1}^{T}\mathbb{P}\left(\|\mathbf{q}_t-\bm{\theta}/\mu\|>\tilde{B}+d\right)\leq D_1 e^{-D_2 d}
\end{equation}
where $\tilde{B}=\!\Theta({1}/{\sqrt{\mu}})$ as in \eqref{eq.the1-4}. 
Intuitively speaking, \eqref{eq.large-dev} upper bounds the probability that the steady-state $\mathbf{q}_t$ deviates from $\bm{\theta}/\mu$, and \eqref{eq.large-dev} implies that the probability that $q^i_t>\theta/\mu+\tilde{B}+d,\;\forall i$ is exponentially decreasing in $D_2d$. 

Using the large deviation bound in \eqref{eq.large-dev}, it follows that
\begin{align}\label{eq.the2-70}
	\mathbf{0}&\stackrel{(f)}{\leq}\! \lim_{T\rightarrow \infty} \frac{1}{T} \sum_{t=1}^{T}\mathbb{E}[-\mathbf{A}\mathbf{x}_t-\mathbf{c}_t]\\
	&\stackrel{(g)}{\leq}\!\lim_{T\rightarrow \infty} \frac{1}{T} \sum_{t=1}^{T}\mathbf{1}\!\cdot\!{M}\,\mathbb{P}\left(\mathbf{q}_t<M\right)\stackrel{(h)}{\leq}\! \mathbf{1}\!\cdot\! MD_1 e^{-D_2(\bm{\theta}/\mu-\tilde{B}-M)}\nonumber
\end{align}
where (f) holds because taking expectation in \eqref{eq.large-dev} over all $d$ implies that the expected queue length is finite [cf. \eqref{eq.probn}], which implies the necessary condition in \eqref{Queue-relax}; (g) follows from \cite[Lemma 4]{huang2014} which establishes that negative accumulated service residual $\sum_{t=1}^{T}\mathbb{E}[\mathbf{A}\mathbf{x}_t+\mathbf{c}_t]$ may happen only when $\mathbf{q}_t\!<\!M$ and the maximum value is bounded by $\|\mathbf{A}\mathbf{x}_t+\mathbf{c}_t\|\leq M$ in Assumption \ref{assp.dualgrad}; and (h) uses the bound in \eqref{eq.large-dev} by choosing $d=\bm{\theta}/\mu-\tilde{B}-M$. 

Setting $\bm{\theta}=\sqrt{\mu}\log^2(\mu)$ in \eqref{eq.the2-70}, there exists a sufficiently small $\mu$ such that $-D_2\big(\log^2(\mu)/\sqrt{\mu}-\tilde{B}-M\big)\leq 2\log(\mu)$. Together with \eqref{eq.the2-70} and $D_1\!=\!\Theta(1/\mu)$, the latter implies that 
\begin{equation}\label{eq.the2-79}
	\left\|\lim_{T\rightarrow \infty}\! ({1}/{T}) \textstyle\sum_{t=1}^{T} \mathbb{E}[-\mathbf{A}\mathbf{x}_t\!-\!\mathbf{c}_t]\right\|\leq \|\mathbf{1}\cdot MD_1\mu^2\|={\cal O}(\mu). 
\end{equation}
Plugging \eqref{eq.the2-79} into \eqref{eq.the2-rho74}, setting $\rho=\mathbf{o}(\mu)$ in \eqref{eq.the2-rho74}, and using $\|\bm{\lambda}^*-\bm{\theta}\|={\cal O}(1)$, we arrive at \eqref{eq.the2-74}.

Letting $T\rightarrow \infty$ in \eqref{eq.the2-71}, it follows from \eqref{eq.the2-72} and \eqref{eq.the2-74} that 
\begin{align}\label{eq.the2-81}
		0&\leq \lim_{T\rightarrow \infty}\frac{1}{T}\!\sum_{t=1}^T\mathbb{E}[{\cal D}(\bm{\gamma}_t)]\!-\!{\cal D}(\bm{\lambda}^*)\!+{\cal O}(\mu)+\!\frac{\mu M^2}{2}\nonumber\\
		&\stackrel{(h)}{\leq} {\cal D}\left(\lim_{T\rightarrow \infty}\frac{1}{T}\!\sum_{t=1}^T\mathbb{E}[\bm{\gamma}_t]\right)\!-\!{\cal D}(\bm{\lambda}^*)\!+{\cal O}(\mu)+\!\frac{\mu M^2}{2}.
\end{align}
where inequality (h) uses the concavity of the dual function ${\cal D}(\bm{\lambda})$.
Defining $\bm{\varphi}:=\lim_{T\rightarrow \infty}\frac{1}{T}\!\sum_{t=1}^T\mathbb{E}[\bm{\gamma}_t]$, and using ${\cal D}(\bm{\lambda}^*)-{\cal D}(\bm{\varphi})\geq \frac{\epsilon}{2}\|\bm{\lambda}^*-\bm{\varphi}\|^2$ in Lemma \ref{lemma.QG}, \eqref{eq.the2-81} implies that
\begin{equation}\label{eq.mu-76}
	\|\bm{\lambda}^*-\bm{\varphi}\|^2\!\leq\!  \frac{2}{\epsilon}\Big({\cal D}(\bm{\lambda}^*)-{\cal D}(\bm{\varphi})\Big)\!\leq\! {\cal O}(\mu)+\frac{\mu M^2}{\epsilon}\!\stackrel{(i)}{=}\!{\cal O}(\mu)\!
\end{equation}
where (i) follows since constants $M$ and $\epsilon$ are independent of $\mu$.
From \eqref{eq.mu-76}, we can further conclude that $\|\bm{\lambda}^*-\bm{\varphi}\|={\cal O}(\sqrt{\mu})$.

Recalling the definition $\bm{\gamma}_t:=\hat{\bm{\lambda}}_t+\mu\mathbf{q}_t-\bm{\theta}$, we  have that
\begin{align}
	&\lim_{T\rightarrow \infty}\!\frac{1}{T}\!\sum_{t=1}^T\mathbb{E}\!\left[\frac{\bm{\gamma}_t}{\mu}\right]\!-\!\frac{\bm{\lambda}^*}{\mu}\!=\!\lim_{T\rightarrow \infty}\!\frac{1}{T}\!\sum_{t=1}^T\mathbb{E}\!\left[\mathbf{q}_t\!+\!\frac{\hat{\bm{\lambda}}_t}{\mu}\right]\!-\!\frac{\bm{\lambda}^*}{\mu}\!-\!\frac{\bm{\theta}}{\mu}\nonumber\\
	\stackrel{(j)}{=}&\lim_{T\rightarrow \infty}\frac{1}{T}\!\sum_{t=1}^T\mathbb{E}\left[\mathbf{q}_t\right]\!-\!\frac{\bm{\theta}}{\mu}\stackrel{(k)}{\leq} \frac{1}{\mu}\|\bm{\lambda}^*-\bm{\varphi}\|={\cal O}\left(\frac{1}{\sqrt{\mu}}\right)
\end{align}
where (j) follows from the convergence of $\hat{\bm{\lambda}}_t$ in Theorem \ref{emp-dual}, and inequality (k) uses the definition of $\bm{\varphi}$ and $\bm{\varphi}-\bm{\lambda}^*\leq \|\bm{\lambda}^*-\bm{\varphi}\|$.
Recalling that $\bm{\theta}=\sqrt{\mu}\log^2(\mu)$ in \eqref{eq.the2-79} completes the proof.

\subsection{Proof of Theorem \ref{gap-onlineLA-SDG}}\label{app.D}
   
Defining the Lyapunov drift as $\Delta(\mathbf{q}_t)\!:=\!\frac{1}{2}(\|\mathbf{q}_{t+1}\|^2\!-\!\|\mathbf{q}_t\|^2)$, and squaring the queue update, we obtain
\begin{align}
	\|\mathbf{q}_{t+1}\|^2=&\|\mathbf{q}_t\|^2+2\mathbf{q}_t^{\top}(\mathbf{A}\mathbf{x}_t+\mathbf{c}_t)+\|\mathbf{A}\mathbf{x}_t+\mathbf{c}_t\|^2\nonumber\\
	\stackrel{(a)}{\leq}&\|\mathbf{q}_t\|^2+2\mathbf{q}_t^{\top}(\mathbf{A}\mathbf{x}_t+\mathbf{c}_t)+M^2
\end{align}
where (a) follows from the definition of $M$ in Assumption \ref{assp.dualgrad}. Multiplying by $\mu/2$ and adding $\Psi_t(\mathbf{x}_t)$, yields
\begin{align}\label{eq.delta-q}
	\mu\Delta(\mathbf{q}_t)+&\Psi_t(\mathbf{x}_t)\leq\Psi_t(\mathbf{x}_t)+\mu\mathbf{q}_t^{\top}(\mathbf{A}\mathbf{x}_t+\mathbf{c}_t)+{\mu M^2}/{2}\nonumber\\
	\stackrel{(b)}{=}&\Psi_t(\mathbf{x}_t)+(\bm{\gamma}_t-\hat{\bm{\lambda}}_t+\bm{\theta})^{\top}(\mathbf{A}\mathbf{x}_t+\mathbf{c}_t)+{\mu M^2}/{2}\nonumber\\
	\stackrel{(c)}{=}&{\cal L}_t(\mathbf{x}_t,\bm{\gamma}_t)+(\bm{\theta}-\hat{\bm{\lambda}}_t)^{\top}(\mathbf{A}\mathbf{x}_t+\mathbf{c}_t)+{\mu M^2}/{2}
\end{align}
where (b) uses the definition of $\bm{\gamma}_t$, and (c) is the definition of the instantaneous Lagrangian. Taking expectations on the both sides of \eqref{eq.delta-q} over $\mathbf{s}_t$ conditioned on $\mathbf{q}_t$, it holds that
\begin{align}\label{eq.57}
	&\mu\mathbb{E}\left[\Delta(\mathbf{q}_t)\big|\mathbf{q}_t\right]+\mathbb{E}\left[\Psi_t(\mathbf{x}_t)\big|\mathbf{q}_t\right]\nonumber\\
	\stackrel{(d)}{=}&{\cal D}(\bm{\gamma}_t)+\mathbb{E}\left[(\bm{\theta}-\hat{\bm{\lambda}}_t)^{\top}(\mathbf{A}\mathbf{x}_t+\mathbf{c}_t)\big|\mathbf{q}_t\right]+{\mu M^2}/{2}\nonumber\\
	  \stackrel{(e)}{\leq}&{\Psi}^{*}+\mathbb{E}\left[(\bm{\theta}-\hat{\bm{\lambda}}_t)^{\top}(\mathbf{A}\mathbf{x}_t+\mathbf{c}_t)\big|\mathbf{q}_t\right]+{\mu M^2}/{2}
\end{align}
where (d) follows from the definition of the dual function \eqref{eq.dual-func}, while (e) uses the weak duality that ${\cal D}(\bm{\gamma}_t)\leq \tilde{\Psi}^{*}$, and the fact that $\tilde{\Psi}^{*}\leq {\Psi}^{*}$ (cf. the discussion after \eqref{eq.reform}). 

Taking expectations on both sides of \eqref{eq.57} over all possible $\mathbf{q}_t$, summing over $t=1,\ldots,T$, dividing by $T$, and letting $T\rightarrow \infty$, we arrive at
\begin{align}\label{eq.optgap}
	\!\!\!&\lim_{T\rightarrow \infty} \frac{1}{T} \sum_{t=1}^{T}\mathbb{E}\left[\Psi_t(\mathbf{x}_t)\right]\nonumber\\
	\!\!\!  \stackrel{(f)}{\leq}&{\Psi}^{*}\!\!+\!\!\lim_{T\rightarrow \infty} \!\frac{1}{T} {\sum_{t=1}^{T}\mathbb{E}\!\left[(\bm{\theta}\!-\!\hat{\bm{\lambda}}_t)\!^{\top}\!(\mathbf{A}\mathbf{x}_t\!+\!\mathbf{c}_t)\right]}\!\!+\!\frac{\mu M^2}{2}\!+\!\!\lim_{T\rightarrow \infty}\!\!\!\frac{\mu\|\mathbf{q}_{1}\|^2}{2T}\nonumber\\
 \!\!\!\stackrel{(g)}{\leq} &{\Psi}^{*}\!+\!\lim_{T\rightarrow \infty} \frac{1}{T} {\sum_{t=1}^{T}\mathbb{E}\left[(\bm{\theta}-\hat{\bm{\lambda}}_t)^{\top}(\mathbf{A}\mathbf{x}_t+\mathbf{c}_t)\right]}\!+\!\frac{\mu M^2}{2}\!
\end{align}
where (f) comes from $\mathbb{E}[\|\mathbf{q}_{T+1}\|^2]\geq 0$, and (g) follows because $\|\mathbf{q}_{1}\|$ is bounded.
One can follow the derivations in \eqref{eq.the2-43}-\eqref{eq.the2-79} to show \eqref{eq.the2-74}, which is the second term in the RHS of \eqref{eq.optgap}. 
Therefore, we have from \eqref{eq.optgap} that
\begin{align}\label{eq.optgap2}
	\lim_{T\rightarrow \infty} \frac{1}{T} \sum_{t=1}^{T}\mathbb{E}\left[\Psi_t(\mathbf{x}_t)\right]& \leq {\Psi}^{*}\!+\!{\cal O}(\mu)+\frac{\mu M^2}{2}
	\end{align}
which completes the proof.

\section*{Acknowledgement}
The authors would like to thank Profs. Xin Wang, Longbo Huang and Jia Liu for
helpful discussions.

\balance

\end{document}